\documentclass[11pt]{article}

\usepackage[a4paper,left=1in,right=1in,top=1in,bottom=1in]{geometry}
\usepackage[english]{babel}
\usepackage[utf8]{inputenc}
\usepackage[T1]{fontenc}
\usepackage{amsmath, amssymb, amsthm}
\usepackage{thmtools, thm-restate}
\usepackage{pifont}
\usepackage{graphicx}
\usepackage{caption}
\usepackage{wrapfig}
\usepackage{syntax}
\usepackage{multicol}
\usepackage{enumitem}
\usepackage{stackengine}
\usepackage{xcolor}
\usepackage{listings}
\usepackage{dirtree}
\usepackage[toc,page]{appendix}
\usepackage{setspace}
\usepackage{url}
\usepackage{hyperref}
\usepackage{microtype}

\lstdefinelanguage{application}{
  morekeywords={<@,@>,db,fun,for,in,do,select,from,where,as,null,\%},
  sensitive=false,
  morestring=[b]",
	alsoletter=<@>,
	basicstyle=\small
}
\hypersetup{
    colorlinks,
    linkcolor={red!50!black},
    citecolor={green!40!black},
    urlcolor={blue!80!black}
}

\newcommand{\ekim}{\stackon[-0.55em]{ê}{\~{}}}

\newcommand{\estr}[1]{\textnormal{\texttt{"#1"}}}
\newcommand{\elet}{\textnormal{\texttt{let}}}
\newcommand{\erec}{\textnormal{\texttt{rec}}}
\newcommand{\ein}{\textnormal{\texttt{in}}}
\newcommand{\elam}{\boldsymbol{\lambda}}
\newcommand{\elcons}{\textnormal{\texttt{cons}}}
\newcommand{\elnil}{\textnormal{\texttt{nil}}}
\newcommand{\etcons}{\textnormal{\texttt{tcons}}}
\newcommand{\etnil}{\textnormal{\texttt{tnil}}}
\newcommand{\eldestr}{\textnormal{\texttt{destr}}}
\newcommand{\etdestr}{\textnormal{\texttt{tdestr}}}
\newcommand{\eif}{\textnormal{\texttt{if}}}
\newcommand{\ethen}{\textnormal{\texttt{then}}}
\newcommand{\eelse}{\textnormal{\texttt{else}}}
\newcommand{\enot}{\textnormal{\texttt{not}}}
\newcommand{\eavg}{\textnormal{\texttt{avg}}}
\newcommand{\esum}{\textnormal{\texttt{sum}}}
\newcommand{\eand}{\textnormal{\texttt{and}}}
\newcommand{\eor}{\textnormal{\texttt{or}}}
\newcommand{\etrue}{\textnormal{\texttt{true}}}
\newcommand{\efalse}{\textnormal{\texttt{false}}}
\newcommand{\edataref}[1]{\textnormal{\texttt{db.}#1}}
\newcommand{\etruffle}[1]{\textnormal{\texttt{truffle<}#1\texttt{>}}}

\newcommand{\oscan}{\textnormal{\textbf{Scan}}}
\newcommand{\oselect}{\textnormal{\textbf{Select}}}
\newcommand{\oproject}{\textnormal{\textbf{Project}}}
\newcommand{\osort}{\textnormal{\textbf{Sort}}}
\newcommand{\otopk}{\textnormal{\textbf{Limit}}}
\newcommand{\ogroup}{\textnormal{\textbf{Group}}}
\newcommand{\ojoin}{\textnormal{\textbf{Join}}}

\newtheorem{definition}{Definition}

\newtheorem{theorem}{Theorem}

\newtheoremstyle{exmp}
  {\topsep}{2\topsep}%
  { }{ }%
  {\bfseries}{ }%
  {\newline}{Example \arabic{example}.}%
\theoremstyle{exmp}
\newtheorem{example}{Example}

\newcommand{\multistepreductionexample}[2]{
{\small
\setlength{\tabcolsep}{0.01\linewidth}
\begin{tabular}{|l@{}@{}c@{}@{}l|} \hline
\parbox[b]{0.47\linewidth}{\centering $e$}
&
\parbox{0.04\linewidth}{\centering $\rightarrow^*$}
&
\parbox[b]{0.47\linewidth}{\centering $e'$}\\ \hline
\parbox{0.47\linewidth}{\vspace{0.2\baselineskip}#1\vspace{0.2\baselineskip}}
&
\parbox{0.04\linewidth}{\vspace{0.2\baselineskip}\centering $\rightarrow^*$\vspace{0.2\baselineskip}}
& 
\parbox{0.47\linewidth}{\vspace{0.2\baselineskip}#2\vspace{0.2\baselineskip}}\\ \hline
\end{tabular}}%
\vspace{-0.5em}}

\newcommand{\onestepreductionexample}[2]{
{\small
\setlength{\tabcolsep}{0.01\linewidth}
\begin{tabular}{|l@{}@{}c@{}@{}l|} \hline
\parbox[b]{0.47\linewidth}{\centering $e$}
&
\parbox{0.04\linewidth}{\centering $\rightarrow$}
&
\parbox[b]{0.47\linewidth}{\centering $e'$}\\ \hline
\parbox{0.47\linewidth}{\vspace{0.2\baselineskip}#1\vspace{0.2\baselineskip}}
&
\parbox{0.04\linewidth}{\vspace{0.2\baselineskip}\centering $\rightarrow$\vspace{0.2\baselineskip}}
& 
\parbox{0.47\linewidth}{\vspace{0.2\baselineskip}#2\vspace{0.2\baselineskip}}\\ \hline
\end{tabular}}%
\vspace{-0.5em}}

\newcommand{\op}[1]{\textsf{Op}\left(#1\right)}
\newcommand{\comp}[1]{\textsf{Comp}\left(#1\right)}
\newcommand{\frag}[1]{\textsf{Frag}\left(#1\right)}

\newcommand{\redexes}[1]{\textsf{Rdxs}\left(#1\right)}
\newcommand{\reduceredex}[1]{\rightarrow_{#1}}
\newcommand{\reduction}[2]{\textsf{Red}_{#1}\left(#2\right)}
\newcommand{\reductions}[1]{\textsf{Reds}\left(#1\right)}
\newcommand{\minreductions}[1]{\textsf{MinReds}\left(#1\right)}
\newcommand{\oppaths}[1]{\textsf{OpContexts}\left(#1\right)}
\newcommand{\cfgfreevars}[2]{\textsf{ConfigFVars}\left(#1,#2\right)}
\newcommand{\cfgredexes}[2]{\textsf{ConfigRdxs}\left(#1,#2\right)}
\newcommand{\opchildren}[2]{\textsf{Children}\left(#1,#2\right)}
\newcommand{\listcomp}[2]{\textnormal{\textsf{[}}#1\ |\ #2\textnormal{\textsf{]}}}
\newcommand{\hredcfg}[3]{\textsf{HRedConfig}_{#1}\left(#2, #3\right)}
\newcommand{\hredchild}[3]{\textsf{HRedChild}_{#1}\left(#2, #3\right)}
\newcommand{\honeredcfg}[2]{\textsf{HOneRedConfig}\left(#1, #2\right)}
\newcommand{\honeredchild}[2]{\textsf{HOneRedChild}\left(#1, #2\right)}
\newcommand{\hinlinevar}[2]{\textsf{HInlineVar}\left(#1, #2\right)}
\newcommand{\hcontractlam}[2]{\textsf{HContractLam}\left(#1, #2\right)}
\newcommand{\hmakeredex}[2]{\textsf{HMakeRedex}\left(#1, #2\right)}
\newcommand{\hnone}{\textsf{None}}
\newcommand{\hconfigss}[2]{\textsf{HConfigSS}_{#1}\left(#2\right)}
\newcommand{\hconfig}[2]{\textsf{HConfig}_{#1}\left(#2\right)}
\newcommand{\hchildss}[2]{\textsf{HChildrenSS}_{#1}\left(#2\right)}
\newcommand{\hchild}[2]{\textsf{HChildren}_{#1}\left(#2\right)}
\newcommand{\hminred}[2]{\textsf{HMinRed}_{#1}\left(#2\right)}
\newcommand{\travpos}[2]{\textsf{TraversalPos}\left(#1,#2\right)}
\newcommand{\starify}[1]{\star\hspace{-0.1em}\left(#1\right)}
\newcommand{\unstarify}[1]{\star^{-1}\hspace{-0.2em}\left(#1\right)}
\newcommand{\starifysub}[1]{\star_{\operatorname{sub}}\hspace{-0.1em}\left(#1\right)}
\newcommand{\unstarifysub}[1]{\star_{\operatorname{sub}}^{-1}\hspace{-0.2em}\left(#1\right)}
\newcommand{\exprid}[2]{\textsf{ExprId}\left(#1,#2\right)}
\newcommand{\occur}[2]{\textsf{Occurences}\left(#1,#2\right)}

\newcommand{\cmark}{\checkmark}

\raggedcolumns
\begingroup\expandafter\expandafter\expandafter\endgroup
\expandafter\ifx\csname pdfsuppresswarningpagegroup\endcsname\relax
\else
\fi

\setlist[itemize]{itemsep=0pt, topsep=0pt, parsep=0pt}

\vspace{-1em}
\title{\textbf{Design of an intermediate representation for query languages}\\ \vspace{1em} Internship report}
\author{Romain Vernoux, supervised by Kim Nguy\ekim n (VALS team, LRI)}
\date{\today}

\begin{document}

\pagestyle{empty}
\pagenumbering{gobble}

\maketitle
\begin{spacing}{0.95}
\tableofcontents
\end{spacing}

\cleardoublepage
\pagestyle{plain}
\pagenumbering{arabic}

\section{Introduction}

This internship is part of a joint collaboration between Laurent Daynes at Oracle Labs, Giuseppe Castagna and Julien Lopez (also a MPRI student) at PPS and Kim Nguy\ekim n and myself at LRI. Each of these three groups worked on different aspects of the project and synchronized during regular meetings. In this report, I will do my best to focus on my contribution and explain only the necessary of the other parts of the project. I will use ``I'' to emphasize my work and ``we'' in formal definitions and technical sections.

\subsection*{The general context}

Data oriented applications, usually written in a high-level, general-purpose programming language (such as Java, Ruby, or JavaScript) interact with databases through a very coarse interface. Informally, the text of a query is built on the application side (either via plain string concatenation or through an abstract notion of statement) and shipped to the database. After its evaluation, the results are then serialized and sent back to the ``application-code'' where they are translated in the application language datatypes. Such roundtrips represent a significant part of the application running time. Moreover, this programming model prevents one from using richer query constructs, such as user-defined functions (UDF). UDFs are functions defined by the application developer, most of the time in a language more expressive than the query language, and used in a query (e.g., in a filter condition of a SQL query). While some databases also possess a ``server-side'' language for this purpose (e.g., PL\slash SQL in Oracle Database and PL\slash pgSQL in PostgreSQL), its integration with the very-optimized query execution engine is still minimal, resulting in notoriously slow query evaluations. The alternative is to evaluate UDFs in the application runtime, leading to additional application-database roundtrips and even poorer performance\cite{ferry}\cite{data-access}. These problems, often refered to as \emph{language-integrated query} issues, are in no way specific to relational databases and also affect the so-called NoSQL databases, which now all provide some form of declarative query language.\bigskip

In this setting, Oracle Labs is developing Truffle\cite{graal}, a framework for representing dynamic languages programs and functions as abstract syntax tree (AST) nodes with particular evaluation rules, and Graal\cite{graal}, a just-in-time (JIT) compiler written in Java, leveraging Oracle's experience from the Java JIT compiler and achieving high performance in evaluating Truffle nodes. A proof of concept, in the form of a full-blown JavaScript runtime prototype written entirely in Truffle and open-source projects for Ruby, R and Python runtimes\cite{graal}, is already available. In principle, this framework could enable to efficiently evaluate UDFs written in high-level dynamic languages directly inside a database embedding a Java Virtual Machine (such as Oracle DB\cite{oracledb}, Cassandra\cite{cassandra} or Hive\cite{Hive}), thus bridging the aforementioned impedence mismatch gap, considerably decreasing the cost of roundtrips between query and UDFs evaluation, and providing the right level of abstraction for a deep integration in the database query optimization process.

\subsection*{The research problem}

Language-integrated queries and impedence mismatch were already research topics in the 1980s\cite{impedence-mismatch}, although they have received increased attention with the popularization of web applications. Various solutions have already been proposed in the past\cite{language-state-of-art} and constant demand for this matter has for instance lead Microsoft to include LINQ\cite{linq}--a framework integrating Microsoft SQL Server querying primitives in C\#, among other languages--in .NET 3.5, and to continually add new features in the subsequent releases.\bigskip

In our opinion, the imperfections of the existing solutions and the new perspectives opened by the Truffle/Graal project justified a reinvestigation of this (still trending) topic. Furthermore, taking in consideration the wide adoption of NoSQL databases, special care should be taken to propose a solution that is not specific to relational databases, and to the extent of our knowledge such work has not been done yet. Thus, the problematic of the QIR project is articulated around the following questions:
\vspace{-0.6em}
\begin{enumerate} \itemsep0pt \parskip0pt \parsep0pt
	\item How to draw the line between application code that should be evaluated in the application runtime, in the database native query processor and in the application language runtime embedded in the database?
	\item How to represent queries in a way that is agnostic of the application language (for reusability) and the target database (to support relational and non-relational ones)?
	\item How to rewrite this query representation to present an efficient query to the database?
	\item How to translate this query representation to the database native language?
	\item How to evaluate application (Truffle) code inside the database?
\end{enumerate}
\vspace{-0.6em}
Although I took part in discussions revolving around all these questions (cf. Section \ref{sec:architecture}), my internship was centered on questions 2 and 3. Questions 1 and 4 were investigated at PPS and question 5 at Oracle Labs.

\subsection*{My contribution}

A review of existing language-integrated query frameworks (cf. Appendix \ref{ap:resources}) and results from a previous internship\cite{sql++} highlighted that existing database query languages (including SQL) share high-level querying primitives (e.g., filtering, joins, aggregation) that can be represented by operators, but differ widely regarding the semantics of their expression language.\bigskip

In order to represent queries in an application language- and database-agnostic manner (question 2 above), I designed a small calculus, dubbed ``QIR'' for \emph{Query Intermediate Representation}, expressive enough to capture querying capabilities offered by mainstream languages. QIR contains expressions, corresponding to a small extension of the pure lambda-calculus, and operators to represent usual querying primitives (Section \ref{sec:qir-constructs}).\bigskip

In the effort to send efficient queries to the database (question 3 above), I abstracted the idea of ``good'' query representations in a measure on QIR terms. Then, I designed an evaluation strategy rewriting QIR query representations into ``better'' ones (Section \ref{sec:qir-evaluation}).\bigskip

\subsection*{Arguments supporting its validity}

As an abstraction layer between application languages and databases, QIR guarantees the robustness of the entire framework regarding changes in these two tiers. But to provide formal evidence of the relevance of my solutions, I wrote an optimality proof for my evaluation strategy on a particular subset of QIR terms, with respect to the aforementioned measure (Theorem \ref{th:completeness}).\bigskip

Additionally, I implemented a prototype of the QIR evaluation strategy and ran experiments showing that (i) the measure captures well the idea of ``good'' queries, (ii) rewritings have a low cost while enabling considerable speedup opportunities for the database and (iii) the evaluation strategy outputs optimal results well outside the scope of the QIR subset captured by the optimality proof (Section \ref{sec:concrete-examples}).

\subsection*{Summary and future work}

The ability to evaluate snippets from the application code inside the database opened new perspectives on the language-intergrated query problems. My approach allows a complete separation between application language and target database considerations, and to this extent is an interesting contribution even outside the scope of our project. Future work includes testing this concept with more combinations of application languages and databases, and observing the impact of such an architecture on real-world web applications.
\section{Notations and conventions}
\label{sec:notations}

\paragraph{Trees} In this document, we will use an infix parenthesis notation for trees. The tree leaves are represented by their name and a tree rooted at a node $N$ with $n$ children $S_1,\ldots,S_n$ is denoted by $N(r_1,\ldots,r_n)$ where $r_i$ is the representation of the subtree rooted at node $S_i$. For instance, if a tree $T$ is composed of a node $N$ which has two children $P$ and $Q$ such that $Q$ has one child $R$, we will write $T$ as $N(P,Q(R))$. 

\paragraph{Contexts} Contexts are trees in which some subtrees have been removed and replaced by ``holes''. A hole is denoted by $[]$. For instance, a context $C$ can be obtained by removing the subtree rooted at $Q$ in $T$ and will be denoted by $N(P,[])$. Context holes can be filled with arbitrary subtrees. If a context $C$ contains one hole, $C[S]$ corresponds to the tree $C$ where the hole has been replaced by the tree $S$. For instance, if $C = N(P, [])$ then $T = C[Q(R)]$. This definition generalizes to contexts with $n$ holes using the $C[S_1,\ldots,S_n]$ notation. Notice that in this case, the context holes are ordered using a prefix (depth-first, left-to-right) traversal of the tree, which corresponds to the left-to-right order while reading the tree notation we use in this document. For instance, if $C = N([], Q([]))$ then $T = C[P,R]$.

\paragraph{List comprehensions} Similarly to sets and the $\{ x\ |\ P \}$ notation for set comprehensions, we use the $\listcomp{x}{P}$ notation for list comprehensions. List comprehensions are order-preserving, that is, a list $\listcomp{f(x)}{x \in L}$ respects the order of $L$ if $L$ is a list. In this document, we will use them to build the list of a node's children in a tree. For instance, $0(\listcomp{x+1}{x \in [1,2,3]})$ will stand for the tree $0(2,3,4)$. A node with an empty children list will be considered to be a leaf.

\paragraph{Proofs} For space reasons, some proofs have been moved to the Appendices. When no proof is given directly below a lemma or a theorem, it can be found in Appendix \ref{ap:proofs}.
\section{Architecture}
\label{sec:architecture}

\begin{wrapfigure}[10]{l}{0.63\textwidth}
\begin{center}
   \raisebox{0pt}[\dimexpr\height-2.3em\relax]{\includegraphics[width=0.9\linewidth]{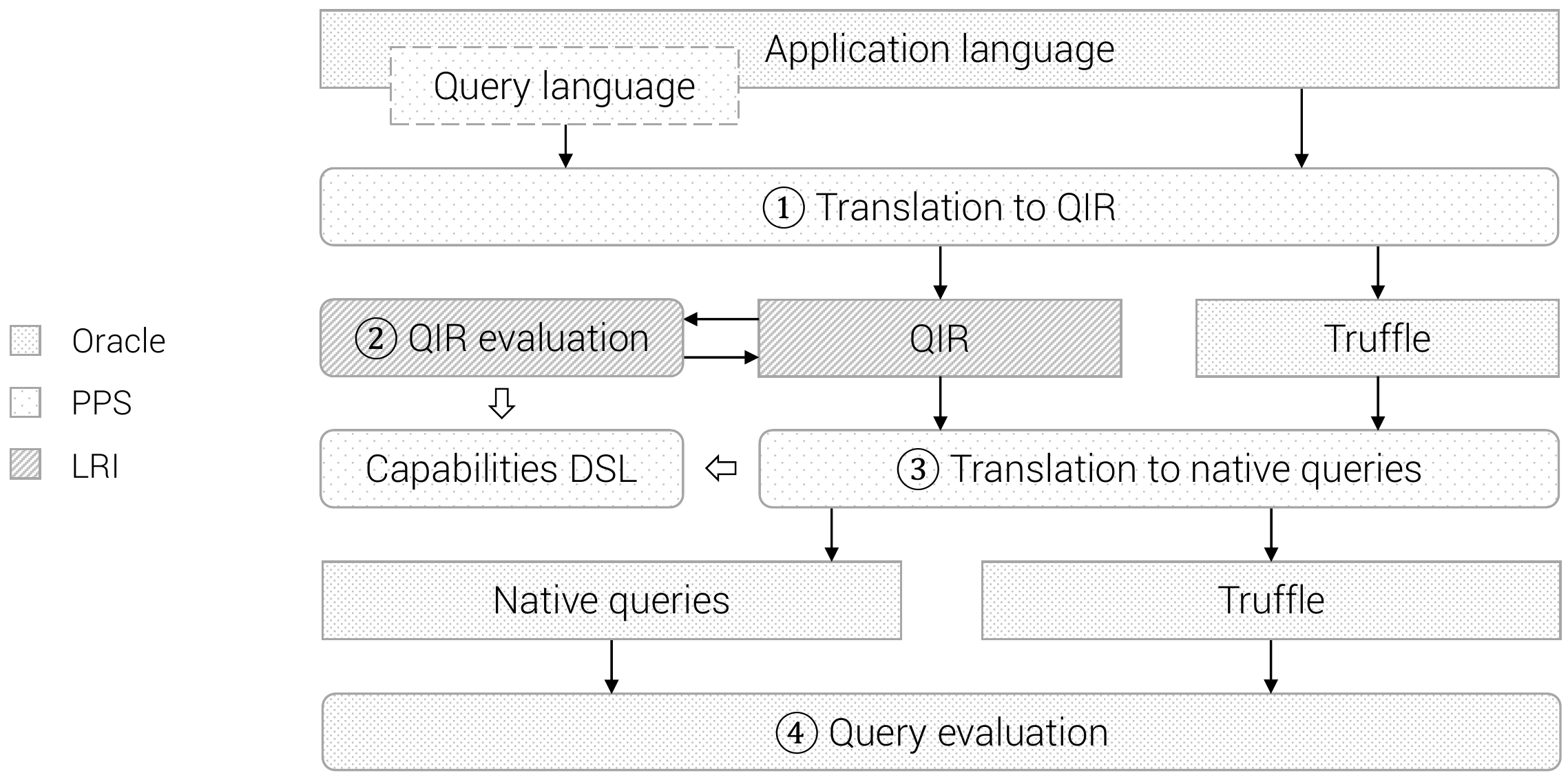}}
	 \caption{Architecture}
	 \label{fig:architecture}
	 \vspace{-2em}
\end{center}
\end{wrapfigure}

In this Section, I present the global architecture of the project and quickly describe the purpose of the modules I did not work on.\bigskip

Since the beginning of my internship, I took part in the design of the global architecture of the project, whose current state is illustrated in Figure \ref{fig:architecture}. It works as follows.\bigskip

\begin{enumerate}
	\item Queries specified in the application language (possibly augmented with a special syntax for querying primitives) are detected in the application code. We refer to the result of this detection as the \emph{query boundary} : code outside this boundary will be translated to Truffle nodes and evaluated in the application runtime, whereas code inside this boundary will be evaluated by the database. Inside the boundary, query parts that can be expressed in QIR are mapped to QIR constructs, and the remaining is propagated as Truffle nodes. This translation phase also maps the application language data model to the QIR data model. For now, we make the assumption that a same construct has identical semantics in the application language and in QIR (e.g., tuple navigations such as \texttt{user.name} return the same result in JavaScript and in QIR, and Ruby integers behave like QIR integers), and we allow a lossy translation when this is not the case. The QIR constructs will be described in Section~\ref{sec:qir-constructs}. The translation phase, designed at PPS, was not the point of my internship and will not be further described. To understand the remaining of this document, you can safely assume that query boundaries are specified by the application developer with a quote/antiquote syntax as in \cite{t-linq}. For example, in the following code snippet (written in an imaginary application language in which queries are specified using a SQL-like syntax),
\begin{lstlisting}[language=application]
f(id) := id = 1;
result := null;
<@ %result := select * from db("ads") as t where %f(t.id)) @>;
for t in result do ... 
\end{lstlisting}
the \texttt{<@ @>} quotes delimit the query boundary, whereas the \texttt{\%} antiquote annotate identifiers bound outside the query boundary. During the application execution, the \texttt{f} and \texttt{result} assignments are evaluated in the application runtime, the quoted code is translated to a QIR expression containing a Truffle reference \texttt{\%f} and the dependencies of this quoted code (i.e., the Truffle definition of \texttt{f}) are passed, along with the QIR code, to the next module in order to be evaluated in the database. 
	\item QIR constructs are partially evaluated within the framework. The goal of this rewriting pass is to feed the database with queries that can be easily optimized by its query optimizer. This module relies on a formal description of the operators that the database can natively evaluate, in order to progressively ``improve'' the query representation with local reductions. Providing this description through a well-defined API is the first role of the ``Capabilities DSL'' (domain-specific language), thus maintaining the nice abstraction over the choice of target database. The partial evaluation of QIR expressions will be described in Section~\ref{sec:qir-evaluation}. This Capabilities DSL and description mechanism, designed at PPS, was not the point of my internship and will not be further described. To understand the remaining of this document, you can safely assume that it is an API providing information such as \emph{This database supports filters, but not joins} and \emph{a filter is compatible if the filter condition has the shape $\elam t.\ t.attr = x$ and if $x$ is a constant value, but not if $x$ is a free variable}.
	\item After the rewriting pass, QIR expressions that can be expressed in the database native query language are mapped to native query constructs, and the remaining is translated to Truffle nodes. This translation phase also maps the QIR data model to the database data model. As in step 1, we make the assumption that a same construct has identical semantics in QIR and in the database query language, and we allow a lossy translation when this is not the case. Providing a way to translate compatible QIR operators into constructs of the database query language through a well-defined API is the second role of the Capabilities DSL, thus maintaining again the abstraction over the choice of target database. For the same reason as above, you can simply assume that this API provides a method mapping natively supported operators (according to the Capabilities DSL) to the corresponding query string in the database query language.
	\item Native queries go through the database optimizer before their evaluation in the database query processor. Such queries may contain references to Truffle definitions for three reasons: (i) the query in the application language was refering to an identifier defined outside the query boundary, (ii) the quoted application code contained features that could not be translated into QIR, and (iii) the query representation in QIR contained features that could not be translated into the database query language. In all these cases, the Truffle runtime embedded in the database is called by the query processor to evaluate the code corresponding to these references. The result of the query is then returned to the host language. This backward translation maps database native data values to the QIR values (using again the Capabilities DSL translation API), then QIR values to the application language data model. Again, this evaluation mechanism, designed at Oracle Labs, was not the point of my internship and will not be further described in this document.
\end{enumerate}

This architecture has the advantage of interleaving an abstraction layer, composed of QIR and Truffle, between the application layer and the database layer. Beyond added cleanliness, this means that no work has to be done on existing database layer modules to support a new application language, and no work has to be done on existing application layer modules to support a new database.
\section{QIR calculus}
\label{sec:qir-constructs}

For the architecture described in Section \ref{sec:architecture} to work, QIR constructs have to be simple enough to be mapped to as many mainstream application language constructs and database language constructs as possible, but also expressive enough to represent interesting queries.\bigskip

To design a QIR calculus satisfying these requirements, I first reviewed existing language-integrated query frameworks (cf. Appendix \ref{ap:resources}), work that I had for the most part already done in a previous internship on a similar area\cite{sql++}. I observed that existing database query languages all share high-level querying primitives (e.g., filtering, joins, aggregation), even though they support these features to various degrees. Moreover, the corresponding language constructs are almost always only a thin language layer hiding to various extents an algebra of operators inspired by the relational algebra. I also noticed the tension between database languages with a declarative descriptions of queries and application languages in which queries are built imperatively.\bigskip

Therefore, in order to represent queries in a manner agnostic of the application and database languages, I designed QIR as a small extension of the pure lambda-calculus, made of (i) \emph{expressions}, i.e., variables, lambdas, applications, constants, functions (conditionals, arithmetic operations, etc.) as well as list/tuple constructors and destructors, and (ii) \emph{operators}, i.e., constructors representing the common querying primitives mentioned above. Thus, queries built imperatively in the application code can easily be translated to expressions manipulating operators whereas declarative queries correspond more to trees of operators, and one can go from the first representation to the second by partially evaluating the corresponding QIR term. The remaining of this section is a formal description of the QIR calculus.

\subsection{QIR data model}
\label{sec:data-model}

QIR data objects are called \emph{values}. The following grammar describes their syntax (for future use in this document) and hierarchy.
\vspace{-0.5em}

\setlength{\grammarindent}{6em}
\begin{multicols}{2}
\begin{grammar}
\small
<Value> ::= <Number>
\alt <String>
\alt <Bool>

<Number> ::= <64-bit IEEE 754>

<String> ::= <UTF-16>

<Bool> ::= \etrue
\alt \efalse
\end{grammar}
\end{multicols}

For now, the QIR data model only contains the necessary for a proof of concept. It will have to be extended with more primitive types for our framework to be used in real-world applications.

\subsection{QIR operators}
\label{sec:operators}

Operators represent computation on database tables. Each operator has \emph{children} expressions (in parentheses) and \emph{configuration} expressions (in subscript). As in the relational algrebra, children provide the operator input tables whereas configurations describe the computation on table elements. The following grammar describes their syntax (for future use in this document) and hierarchy.

\setlength{\grammarindent}{6.5em}
\begin{multicols}{2}
\begin{grammar}
\small
<Operator> ::= $\oscan_{<Expr>}()$
\alt $\oselect_{<Expr>}(<Expr>)$
\alt $\oproject_{<Expr>}(<Expr>)$
\alt $\osort_{<Expr>}(<Expr>)$
\alt $\otopk_{<Expr>}(<Expr>)$
\alt $\ogroup_{<Expr>,<Expr>}(<Expr>)$
\alt $\ojoin_{<Expr>}(<Expr>,<Expr>)$
\end{grammar}
\end{multicols}

These operators represent the following database computations:
\begin{itemize}
	\item $\oscan_{table}()$ outputs the list (in a non-deterministic order) of elements in the target database table denoted by the expression $table$.
	\item $\oselect_{filter}(input)$ outputs the list of elements $v$ in the list corresponding to $input$ such that $filter\ v$ reduces to $\etrue$.
	\item $\oproject_{format}(input)$ outputs the list of elements corresponding to $format\ v$, with $v$ ranging in the list corresponding to $input$.
	\item $\osort_{comp}(input)$ outputs the list of elements $v$ in the list corresponding to $input$ ordered according to $comp\ v$ ascending.
	\item $\otopk_{limit}(input)$ outputs the $limit$ first elements of the list corresponding to $input$.
	\item $\ogroup_{eq, agg}(input)$ outputs the list of elements corresponding to $agg\ g$ for each group $g$ in the partition of the elements $v$ in the list corresponding to $input$, according to $eq\ v$.
	\item $\ojoin_{filter}(input_1, input_2)$ outputs the join of the elements $v_1$ in the list corresponding to $input_1$ and the elements $v_2$ in the list corresponding to $input_2$, such that $filter\ v_1\ v_2$ reduces to $\etrue$.
\end{itemize}
\vspace{0.8em}

For the same reason as above, QIR operators are still limited. Nervertheless, they already capture a consequent part of the relational algebra.

\subsection{QIR expressions}
\label{sec:expressions}

QIR expressions are the core of the QIR calculus. The following grammar describes their syntax (for future use in this document) and hierarchy.

\setlength{\grammarindent}{7.5em}
\begin{multicols}{2}
\begin{grammar}
\small
<Expr> ::= <Variable>
\alt <Lambda>
\alt <Application>
\alt <Constant>
\alt <ValueCons>
\alt <ValueDestr>
\alt <ValueFun>
\alt <BuiltinFun>
\alt <DataRef>
\alt <TruffleNode>
\alt <Operator>

<Lambda> ::= \lit{$\elam$} <Variable> `.' <Expr>

<Application> ::= <Expr> <Expr>

<Constant> ::= <Value>

<ValueCons> ::= <ListCons>
\alt <TupleCons>

<ListCons> ::= \lit{$\elnil$}
\alt \lit{$\elcons$} <Expr> <Expr>

<TupleCons> ::= \lit{$\etnil$}
\alt \lit{$\etcons$} <String> <Expr> <Expr>

<ValueDestr> ::= <ListDestr>
\alt <TupleDestr>

<ListDestr> ::= \lit{$\eldestr$} <Expr> <Expr> <Expr>

<TupleDestr> ::= \lit{$\etdestr$} <Expr> <String>

<ValueFun> ::= `-' <Expr>
\alt <Expr> \lit{$\eand$} <Expr>
\alt \lit{$\eif$} <Expr> \lit{$\ethen$} <Expr> \lit{$\eelse$} <Expr>
\alt \ldots

<BuiltinFun> ::= \lit{$\eavg$} <Expr>
\alt \ldots

<DataRef> ::= \lit{$\edataref{table\_name}$}

<TruffleNode> ::= \lit{$\etruffle{id}$}
\end{grammar}
\end{multicols}

Lambda-abstractions (resp. applications, constants) represent functions (resp. function applications, constants) of the application language. The constructor for lists takes an expression for the head and an expression for the tail. Tables are represented by lists, since the output of queries might be ordered. Tuples are constructed as a list of mappings: the constructor for tuples takes a string for the mapping name, an expression for the mapping value and an expression for the tail of the mapping list. The list destructor has three arguments: the list to destruct, the term to return when the list is \texttt{nil} and a function with two arguments $\elam h. \elam t. M$ to treat the case when the list has a head and a tail. The tuple destructor has two arguments: the tuple to destruct and the attribute name of the value to return. Finally, built-in functions represent common database functions (such as aggregation) for easier recognition and translation into native languages. The formal definition of expression reductions is given in Section \ref{sec:qir-evaluation}.\bigskip

In this document, we will sometimes use the syntactic sugar $\elet \ x = M \ \ein \ N$ for $(\elam x.\ N) M$ but the $\elet \ldots \ein$ constructor is not part of the expression language. Similarly, we will sometimes use the notation $t.attr$ for $\etdestr \ t \ \estr{attr}$, $\enot\ x$ for $\eif\ x\ \ethen\ \efalse\ \eelse\ \etrue$ and $\elet \ \erec \ x = M \ \ein \ N$ as a shortcut for the application of a fixpoint combinator.
\section{QIR evaluation}
\label{sec:qir-evaluation}

As discussed in the beginning of Section \ref{sec:qir-constructs}, the purpose of the QIR partial evaluation mechanism is to transform the QIR representation of a query written in the (possibly imperative) style of the application language into a QIR representation of the same query that is easier to translate to the database native language and offers more optimization opportunities to its optimizer. Today's databases use a tree of algebra operators (plan) to represent computation. Optimizations consist in commuting operators in the plan or applying local tree rewritings. The plan evaluation usually corresponds to (a variation of) a bottom-up evaluation of each operator, where a particular implementation is chosen for each operator for best performance. Therefore, for a same computation, such database engines benefit from working on a large, single plan instead of multiple small plans using the result of each other, since (i) more optimizations become available, (ii) more statistics can be gathered to pick the best implementation for each operator and (iii) no intermediate result needs to be materialized and transferred.\bigskip

Consider the following example of query (written in an imaginary application language in which queries are specified using a SQL-like syntax).

\begin{lstlisting}[language=application]
result := null;
<@ users := db("users"); ads := db("ads");
filter(x,y) := (x.id = y.userid) and (y.timestamp > 1234);
prettify(s) := ...;
%result = select prettify(name) as u_name, descr as a_descr
          from users as u, ads as a where filter(u,a) @>
for t in result do ...
\end{lstlisting}

The quoted code contains a query as well as the definition of two UDFs: \texttt{filter}, which is a quite simple filter function and \texttt{prettify}, which contains application languages features that cannot be translated to QIR. A direct translation of the quoted code to QIR would be as follows.

\begin{lstlisting}[basicstyle=\small]
$\elet\ users = \edataref{users}\ \ein\ \elet\ ads = \edataref{ads}\ \ein$
$\elet\ filter = \elam x.\ \elam y.\ (\etdestr\ x\ \estr{id}) = (\etdestr\ y\ \estr{userid})\ \eand\ (\etdestr\ y\ \estr{timestamp} > 1234)\ \ein$
$\elet\ prettify = \etruffle{0}\ \ein$
$\oproject_{\elam t.\ \etcons\ \estr{u\_name}\ (prettify\ (\etdestr\ t\ \estr{name}))\ (\etcons\ \estr{a\_descr}\ (\etdestr\ t\ \estr{descr})\ \etnil)}($
  $\ojoin_{\elam u.\ \elam a.\ filter\ u\ a}(\oscan_{users}(),\oscan_{ads}()))$
\end{lstlisting}

Notice that the definition of \texttt{prettify} had to be translated to Truffle nodes (not shown here) and is represented in the QIR code by a Truffle reference. Assuming that the target database can only natively evaluate operators (e.g., SQL), the direct translation of the above QIR code to the database language would lead to a query in which all operators issue calls to UDFs, greatly impacting performances. Moreover, with no information about \texttt{filter}, the database optimizer could not decide to push the selection \texttt{timestamp > 1234} below the join or choose an index-based implementation for the join. In fact, a better QIR representation for this query is as follows.

\begin{lstlisting}[basicstyle=\small]
$\oproject_{\elam t.\ \etcons\ \estr{u\_name}\ (\etruffle{0}\ (\etdestr\ t\ \estr{name}))\ (\etcons\ \estr{a\_descr}\ (\etdestr\ t\ \estr{descr})\ \etnil)}($
  $\ojoin_{\elam u.\ \elam a.\ (\etdestr\ u\ \estr{id})\ =\ (\etdestr\ a\ \estr{userid})\ \eand\ (\etdestr\ a\ \estr{timestamp} > 1234)}(\oscan_{\edataref{users}}(),\oscan_{\edataref{ads}}()))$
\end{lstlisting}

This second version can be obtained from the first by a reduction of the QIR code. As we will see later in this section, applying classical reduction strategies (e.g., call-by-value, call-by-name, lazy evaluation) until a normal form is reached does not satisfy our requirements, since contracting some redexes might scatter parts of a same query. Instead, we will characterize the shape of ``good'' QIR code with a measure on QIR expressions in Section \ref{sec:measure}, then investigate reduction strategies guided by this measure in Sections \ref{sec:exhaustive} and \ref{sec:heuristic}.

\subsection{QIR reduction rules}
\label{sec:reduction-rules}

In this section, I formally define the reduction of QIR expressions, then state two useful properties on the QIR calculus equipped with this reduction.

\begin{definition}
\label{def:substitution}
The capture-avoiding variable substitution in QIR expressions (cf. Section \ref{sec:expressions}) is defined as in the pure lambda-calculus. The substitution of a variable $x$ in $e_1$ by an expression $e_2$ is denoted $e_1\{e_2/x\}$.
\end{definition}

\begin{definition}
\label{def:reduction-rules}
The reduction rules for QIR expressions consist in the $\beta$-reduction rule augmented with the $\delta$-reduction rules for destructors and $\rho$-reduction rules for primitive types.\smallskip

{\small
\setlength{\tabcolsep}{0.01\linewidth}%
\noindent\begin{tabular}{ l c l }
\parbox{0.45\linewidth}{$(\elam x.\ e_1)\ e_2$}
&
\parbox{0.04\linewidth}{\centering $\rightarrow_{\beta}$}
&
\parbox{0.45\linewidth}{$e_1\{e_2/x\}$}
\end{tabular}

\vspace{1em}

\noindent\begin{tabular}{ l c l }%
\parbox{0.45\linewidth}{$\eldestr \ \elnil \ e_{nil} \ e_{cons}$}
&
\parbox{0.04\linewidth}{\centering $\rightarrow_{\delta}$}
&
\parbox{0.45\linewidth}{$e_{nil}$}\\
\parbox{0.45\linewidth}{$\eldestr \ (\elcons \ e_{head} \ e_{tail}) \ e_{nil} \ e_{cons}$}
&
\parbox{0.04\linewidth}{\centering $\rightarrow_{\delta}$}
&
\parbox{0.45\linewidth}{$e_{cons} \ e_{head} \ e_{tail}$}\\
\parbox{0.45\linewidth}{$\etdestr \ (\etcons \ \estr{name1} \ e_{val1} \ e_{tail}) \ \estr{name1}$}
&
\parbox{0.04\linewidth}{\centering $\rightarrow_{\delta}$}
&
\parbox{0.45\linewidth}{$e_{val1}$}\\
\parbox{0.45\linewidth}{$\etdestr \ (\etcons \ \estr{name1} \ e_{val1} \ e_{tail}) \ \estr{name2}$}
&
\parbox{0.04\linewidth}{\centering $\rightarrow_{\delta}$}
&
\parbox{0.45\linewidth}{$\etdestr \ e_{tail} \ \estr{name2}$}
\end{tabular}

\vspace{1em}

\noindent\begin{tabular}{ l c l }%
\parbox{0.45\linewidth}{$\eif\ \etrue\ \ethen\ e_1\ \eelse\ e_2$}
&
\parbox{0.04\linewidth}{\centering $\rightarrow_{\rho}$}
&
\parbox{0.45\linewidth}{$e_1$}\\
\parbox{0.45\linewidth}{$\eif\ \efalse\ \ethen\ e_1\ \eelse\ e_2$}
&
\parbox{0.04\linewidth}{\centering $\rightarrow_{\rho}$}
&
\parbox{0.45\linewidth}{$e_2$}\\
\parbox{0.45\linewidth}{$\etrue \ \eand \ e_1$}
&
\parbox{0.04\linewidth}{\centering $\rightarrow_{\rho}$}
&
\parbox{0.45\linewidth}{$\etrue$}\\
\parbox{0.45\linewidth}{$\efalse \ \eand \ e_1$}
&
\parbox{0.04\linewidth}{\centering $\rightarrow_{\rho}$}
&
\parbox{0.45\linewidth}{$e_1$}\\
\parbox{0.45\linewidth}{$\ldots$}
&
\parbox{0.04\linewidth}{}
&
\parbox{0.45\linewidth}{}
\end{tabular}}

\noindent In this document, we will use the notation $\rightarrow$ for the relation $\rightarrow_{\beta} \cup \rightarrow_{\delta} \cup \rightarrow_{\rho}$ and talk about \emph{redexes} for $\beta$-redexes, $\delta$-redexes and $\rho$-redexes indifferently. Furthermore, we will denote by $\rightarrow^n$ the $n$-step reduction relation and by $\rightarrow^*$ the reflexive, transitive closure of $\rightarrow$.
\end{definition}
\vspace{-0.5em}

\begin{restatable}{theorem}{thconfluence}
\label{th:confluence}
QIR with the reduction rules of Definition \ref{def:reduction-rules} satisfies the Church-Rosser property.
\end{restatable}

As we will see in the following sections, this is an important property, since it allows us to apply any reduction strategy while preserving the semantics of the input expression. Consequently, extending QIR with non-confluent constructs (e.g., side-effects) would require significant work.

\begin{restatable}{theorem}{thstandardization}
\label{th:standardization}
QIR with the reduction rules of Definition \ref{def:reduction-rules} verifies the Standardization theorem, that is, if an expression $e$ has a normal form, it can be obtained from $e$ by reducing successively its leftmost outermost redex.
\end{restatable}

\subsection{A measure for ``good'' output plans}
\label{sec:measure}

In this Section, I formally define the measure characterizing ``good'' QIR expressions. 

\begin{definition}
\label{def:supported-op}
\emph{Supported operators} are operators that are supported by the database, according to the mechanism discussed in Section \ref{sec:architecture}.
\end{definition}

\begin{definition}
\label{def:supported-expr}
\emph{Supported expressions} are expressions that are supported by the database inside operator configurations, according to the mechanism discussed in Section \ref{sec:architecture}.
\end{definition}

\begin{definition}
\label{def:compatible-op}
An operator is called \emph{compatible} if it is supported and has supported expressions as configuration. Any other operator or expression that is not an operator is called \emph{incompatible}.
\end{definition}

Note that the compatibility of an operator is independent of the compatibility of its children.\bigskip

\noindent\begin{minipage}[l]{0.54\linewidth}
\begin{definition}
\label{def:fragment}
Let $e$ be an expression and $F$ a context. We say that $F$ is a \emph{fragment} if either $e = C[t(t_1,\ldots,t_{i-1},F[e_1,\ldots,e_n],t_{i+1},\ldots,t_j)]$ or $e = F[e_1,\ldots,e_n]$, where:
\begin{itemize}
	\item $C$ is a one-hole context with arbitrary expressions
	\item $t$ is an incompatible j-ary expression
	\item $F$ is a $n$-hole context made only of compatible operators
	\item $t_1,\ldots,t_{i-1},t_{i+1},\ldots,t_j$ and $F[e_1,\ldots,e_n]$ are the $j$ children of $t$
	\item $e_1,\ldots,e_n$ are incompatible expressions
\end{itemize}
\vspace{0.6em}
This definition is illustrated in Figure \ref{fig:fragment}.
\end{definition}
\end{minipage}%
\begin{minipage}[c]{0.46\textwidth}
\centering
\captionsetup{type=figure}
\includegraphics[width=0.95\linewidth]{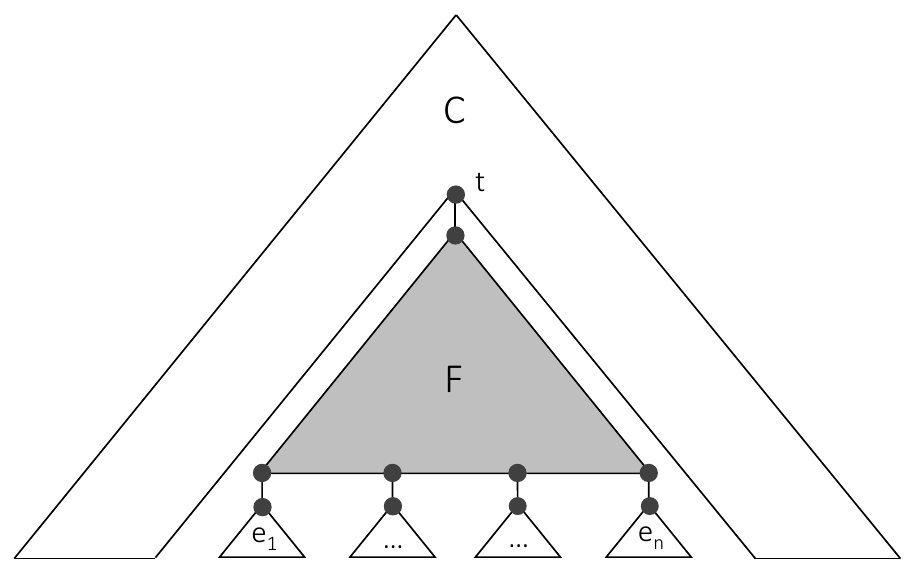}
\captionof{figure}{A fragment $F$}
\label{fig:fragment}
\end{minipage}

\begin{definition}
\label{def:frag-notations}
Let $e$ be an expression. We define $\op{e}$ as the number of operators in $e$, $\comp{e}$ as the number of compatible operators in $e$ and $\frag{e}$ as the number of fragments in $e$.
\end{definition}

\begin{definition}
\label{def:measure}
Let $e$ be an expression. We define the measure for good output expressions as the pair $M(e) = (\op{e} - \comp{e}, \frag{e})$. Moreover, expressions are ordered using the lexicographical order on this measure.
\end{definition}

The intuition about this measure is as follows. We want to consider reduction chains for which the measure decreases, that is, reductions $e \rightarrow^*  e'$ verifying one of the following conditions.
\begin{itemize}
	\item $\op{e'} - \comp{e'} < \op{e} - \comp{e}$ 
	\item $\op{e'} - \comp{e'} = \op{e} - \comp{e}$ and $\frag{e'} < \frag{e}$
\end{itemize}
When a reduction transforms an incompatible operator into a compatible operator, the measure $\op{e} - \comp{e}$ decreases. For instance, in Example \ref{ex:make-compatible}, $\op{e} - \comp{e} = 2 - 1 = 1$ (since $f$ is a variable, it is not supported by the database) and $\op{e'} - \comp{e'} = 2 - 2 = 0$.

\begin{example}
\label{ex:make-compatible}
\multistepreductionexample{
  $\elet \ f = \elam x.\ x = 2 \ \ein\ \oselect_{\elam t. f(t.id)}(\oscan_{e_1}())$
}{
  $\oselect_{\elam t.\ t.id = 2}(\oscan_{e_1}())$
}
\end{example}

\noindent When a reduction inlines operators as children of other operators, $\op{e} - \comp{e}$ stays the same (since no operator is created nor modified) and $\frag{e}$ decreases. For instance, in Examples \ref{ex:inline-fragment1} and \ref{ex:inline-fragment2}, $\op{e} - \comp{e} = \op{e'} - \comp{e'} = 2 - 2 = 0$, $\frag{e} = 2$ and $\frag{e'} = 1$.

\begin{example}
\label{ex:inline-fragment1}
\multistepreductionexample{
  $\elet\ x = \oscan_{e_1}()\ \ein\ \oselect_{\elam t.\ t.id = 2}(x)$
}{
  $\oselect_{\elam t.\ t.id = 2}(\oscan_{e_1}())$
}
\end{example}

\begin{example}
\label{ex:inline-fragment2}
\multistepreductionexample{
  $\elet\ f = \elam x.\ \oselect_{\elam t.\ t.id = 2}(x)\ \ein\ f\ (\oscan_{e_1}())$
}{
  $\oselect_{\elam t.\ t.id = 2}(\oscan_{e_1}())$
}
\end{example}

In the infinite recursion of Example \ref{ex:rec}, unfolding one step of recursion does not change the measure, which hints that the plan is not getting better.

\begin{example}
\label{ex:rec}
\multistepreductionexample{
  $\elet \ \texttt{rec}\ f = \elam x.\ \oselect_{\elam t.\ t.id = 2}(f\ x)\ \ein$\\
  $f\ (\oscan_{e_1}())$
}{
  $\elet \ \texttt{rec}\ f = \elam x.\ \oselect_{\elam t.\ t.id = 2}(f\ x)\ \ein$\\
  $\oselect_{\elam t.\ t.id = 2}(f\ (\oscan_{e_1}()))$
}
\end{example}

This intuition generalizes well to binary operators, with which a reduction can duplicate operators (e.g., Example \ref{ex:duplicate}). If all duplicated operators are compatible, $\op{e} - \comp{e}$ stays the same (since the duplicated operators are compatible) therefore it is a good reduction if the number of fragments decreases. In this case, the database will get a bigger plan in which it can perform optimizations (such as caching). Conversely, if a duplicated operator is incompatible, $\op{e} - \comp{e}$ increases, meaning that the reduction should not happen. In this situation, this operator should in fact be evaluated once in memory and its result cached to be reused multiple times during the expression evaluation. For instance, in Example \ref{ex:duplicate}, if $e_1$ is supported, $\op{e} - \comp{e} = 3 - 3 = 0$, $\op{e'} - \comp{e'} = 5 - 5 = 0$, $\frag{e} = 2$ and $\frag{e'} = 1$, and if $e_1$ is not supported, $\op{e} - \comp{e} = 3 - 2 = 1$, $\op{e'} - \comp{e'} = 5 - 3 = 2$.

\begin{example}
\label{ex:duplicate}
\multistepreductionexample{
  $\elet\ x = \oselect_{e_1}(\oscan_{e_2}())\ \ein$\\
	$\ojoin_{\elam t1. \elam t2.\ t1.id = t2.name}(x,x)$
}{
  $\ojoin_{\elam t1. \elam t2.\ t1.id = t2.name}(\oselect_{e_1}(\oscan_{e_2}()),$\\
	$\oselect_{e_1}(\oscan_{e_2}()))$
}
\end{example}

\begin{restatable}{lemma}{lemwellfounded}
\label{lem:well-founded}
The measure $M$ of Definition \ref{def:measure} induces a well-founded order on expressions.
\end{restatable}

\subsection{An exhaustive reduction strategy}
\label{sec:exhaustive}

Given the measure $M$ described in Section \ref{sec:measure} and an expression $e$, the question is now to find a reduced term $e'$ such that $e \rightarrow^* e'$ and $M(e')$ is minimal. During my internship, I first designed and implemented the exhaustive reduction strategy presented formally in this section.

\begin{definition}
\label{def:redex-contract}
Let $e$ be an expression. We denote by $\redexes{e}$ the set of redexes in $e$, and for $r \in \redexes{e}$, we write $e \reduceredex{r} e'$ to state that $e'$ can be obtained from $e$ in one reduction step by contracting $r$.
\end{definition}

\begin{definition}
\label{def:exhaustive-red}
Let $e$ be an expression. We recursively define $\reduction{0}{e}$ as the singleton $\{e\}$ and $\forall n > 0, \reduction{n}{e} = \{ e'' \ | \ e' \in \reduction{n-1}{e}, r \in \redexes{e'}, e' \reduceredex{r} e''\}$. We write $\reductions{e}$ to denote $\bigcup_{n \in \mathbb{N}} \reduction{n}{e}$. Finally we define $\minreductions{e}$ as the set $\operatorname{argmin}_{e' \in \reductions{e}}M(e').$
\end{definition}

Notice that for some expressions $e$ (e.g., expressions containing recursions), $\reductions{e}$ and $\minreductions{e}$ may be infinite and an algorithm that iteratively constructs the $\reduction{n}{e}$ may not terminate. Conversely, if $e$ is strongly normalizing, $\reductions{e}$ and $\minreductions{e}$ are finite and such an algorithm terminates. In general, the parallel exploration of all redexes leads to a combinatorial explosion, and I will show in Section \ref{sec:concrete-examples} that my implementation of this algorithm timeouts on moderately large expressions.

\begin{theorem}
\label{th:exhaustive}
The reduction strategy of Definition \ref{def:exhaustive-red} is exhaustive, i.e., $\forall e, \forall e', e \rightarrow^* e' \Leftrightarrow e' \in \reductions{e}$.
\begin{proof}
($\Leftarrow$) If $e' \in \reductions{e}$ then $\exists n \in \mathbb{N}, e' \in \reduction{n}{e}$. The proof follows an induction on $n$. If $n = 0$ then $e = e'$ therefore $e \rightarrow^* e'$. For $n > 0$, we know that $\exists e'' \in \reduction{n-1}{e}, \exists r \in \redexes{e''}, e'' \reduceredex{r} e'$. Thus, $e'' \rightarrow e'$ and by induction hypothesis $e \rightarrow^* e''$ therefore $e \rightarrow^* e'$.\\
($\Rightarrow$) If $e \rightarrow^* e'$ then $\exists n \in \mathbb{N}, e \rightarrow^n e'$. The proof follows an induction on $n$. If $n = 0$ then $e = e'$ and $e' \in \reduction{0}{e} \subseteq \reductions{e}$. For $n > 0$, we know that $\exists e'', e \rightarrow^{n-1} e'' \rightarrow e'$. By induction hypothesis, $e'' \in \reductions{e}$ and $\exists n' \in \mathbb{N}, e'' \in \reduction{n'}{e}$. Therefore $e' \in \reduction{n'+1}{e} \subseteq \reductions{e}$.
\end{proof}
\end{theorem}
\vspace{-0.5em}

\begin{restatable}{lemma}{lemredsnonempty}
\label{lem:reds-non-empty}
$\forall e, \minreductions{e} \neq \emptyset$.
\end{restatable}

\begin{theorem}
\label{th:optimal}
The reduction strategy of Definition \ref{def:exhaustive-red} is optimal, i.e., $\forall e, \forall e', e \rightarrow^* e' \Rightarrow e' \in \minreductions{e} \ \vee \ \exists e'' \in \minreductions{e}, M(e'') < M(e')$.
\begin{proof}
Suppose that $e \rightarrow^* e'$. Using Theorem \ref{th:exhaustive} we know that $e' \in \reductions{e}$. Using Lemma \ref{lem:reds-non-empty}, we know that $\minreductions{e} \neq \emptyset$ and we denote by $M_{min}$ the measure $M$ of the elements of $\minreductions{e}$. Then, either $M(e') = M_{min}$ and $e' \in \minreductions{e}$ by definition, or $M(e') > M_{min}$ and using again Lemma \ref{lem:reds-non-empty} we can find $e'' \in \minreductions{e}, M(e'') < M(e')$.
\end{proof}
\end{theorem}

As discussed above, this reduction strategy is not realistic in terms of complexity for a direct implementation and may not even terminate in some cases. Nevertheless, Theorem \ref{th:optimal} states that it can be used as a yardstick to discuss the optimality of other reduction strategies.

\subsection{A heuristic-based reduction strategy}
\label{sec:heuristic}

Given the conclusion of Section \ref{sec:exhaustive}, I started working on a more efficient evaluation strategy. In this section, I formally describe an efficient, always-terminating heuristic, corresponding to a partial exploration of the possible reductions, then give proofs of its properties.\bigskip

This heuristic-based reduction strategy supposes the existence of an integer constant $\Phi$ representing the ``fuel'' that can be consumed by the reduction. It consists of two main passes. The first pass tries to reduce redexes that could make operators compatible by assuming that (i) operators with free variables in their configuration have few chances to be compatible and (ii) reducing redexes inside operator configurations increases the odds of making an operator compatible. The second pass tries to decrease the number of fragments by reducing redexes inside the children expressions of the operators. Both passes guarantee that the measure of the global expression always decreases after a number of reduction steps bounded by $\Phi$.\bigskip

For readability purposes, I first describe the search space tree explored by the heuristic then define the result of the algorithm as a particular expression in this search space. The actual implementation corresponds to a depth-first exploration of the search space, with decision points and back-tracking.

\begin{definition}
\label{def:oppaths}
Let $e$ be an expression. We define its operator contexts $\oppaths{e}$ as the set of contexts $\{ C[]\ |\ e = C[op],\ op \text{ is an operator}\}$.
\end{definition}

\begin{definition}
\label{def:cfgfreevars}
Let $e$ be an expression and $C[]$ an operator context in $\oppaths{e}$. We define the free variables of a configuration $\cfgfreevars{e}{C[]}$ as the list of variables $\listcomp{v}{e = C[op_{C'[v]}(\ldots)],\ v\ \text{is a free variable in } C'[v]}$ sorted using a depth-first left-to-right traversal of $C'[v]$.
\end{definition}

\begin{definition}
\label{def:cfgredexes}
Let $e$ be an expression and $C[]$ an operator context in $\oppaths{e}$. We define the redexes of a configuration $\cfgredexes{e}{C[]}$ as the list of redexes $\listcomp{r}{e = C[op_{C'[r]}(\ldots)],\ r\ \text{is a}\allowbreak \text{redex}}$ sorted using a depth-first left-to-right traversal of $C'[v]$.
\end{definition}

\begin{definition}
\label{def:opchildren}
Let $e$ be an expression and $C[]$ an operator context in $\oppaths{e}$. We define the children expressions $\opchildren{e}{C[]}$ as the list of expressions $\listcomp{c_i}{e = C[op_{\ldots}(c_1,\ldots,c_n)]}$.
\end{definition}

\begin{definition}
\label{def:makeredex-contractlam-inlinevar}
The following three definitions are mutually recursive and therefore presented together in this document.\\
Let $e$ be an expression and $e'$ a subexpression of $e$. We define $\hmakeredex{e}{e'}$ as
\begin{itemize}
	\item $e''$, such that $e \reduceredex{r} e''$, if $e'$ is already a redex $r$
	\item $\hmakeredex{e}{e''}$ if $e' = (e''\ e_1)$, $e' = \eldestr\ e''\ e_1\ e_2$, $e' = \etdestr\ e''\ s$, $e' = \eif\ e''\ \ethen\ e_1\allowbreak \eelse\ e_2$, $e' = e''\ \eand\ (\etrue/ \efalse)$ or $e' = (\etrue/ \efalse)\ \eand\ e''$ (and similarily for other $\rho$-redexes).
	\item $\hinlinevar{e}{v}$ if $e'$ is a variable $v$
	\item $\hnone$ otherwise
\end{itemize}
Let $e$ be an expression and $e'$ a subexpression of $e$. We define $\hcontractlam{e}{e'}$ as
\begin{itemize}
	\item $\hmakeredex{e}{e''}$ where $e = C[e'']$, if $e = C[e'\ e_1]$, $e = C[e_1\ e']$, $e = C[\eldestr\ e'\ e_1\ e_2]$, $e = C[\eldestr\ e_1\ e'\ e_2]$, $e = C[\eldestr\ e_1\ e_2\ e']$, $e = C[\etdestr\ e'\ s]$, $e = C[\eif\ e_1\ \ethen\ e' \eelse\ e_2]$ or $e = C[\eif\ e_1\ \ethen\ e_2\ \eelse\ e']$
	\item $\hcontractlam{e}{e''}$ where $e = C[e'']$, if $e = C[\elam v.\ e']$, $e = C[\elcons\ e'\ e_1]$, $e = C[\elcons\ e_1\ e']$, $e = C[\etcons\ s\ e'\ e_1]$ or $e = C[\etcons\ s\ e_1\ e']$
	\item $\hnone$ otherwise
\end{itemize}
Let $e$ be an expression and $v$ a variable in $e$. We define the inlining $\hinlinevar{e}{v}$ as
\begin{itemize}
	\item $\hnone$ if $v$ is free in $e$
	\item $\hcontractlam{e}{l}$, where $l$ is the $\elam$ binding $v$ in $e$, otherwise
\end{itemize}
\end{definition}

$\hmakeredex{e}{e'}$ corresponds to (i) contracting $e'$ if $e'$ is already a redex and (ii) contracting a necessary redex in order for $e'$ to become a redex, otherwise. $\hcontractlam{e}{e'}$ corresponds to (i) contracting the redex containing $e'$ if $e'$ is already part of a redex and (ii) contracting a necessary redex in order for $e'$ to become part of a redex, otherwise. $\hinlinevar{e}{v}$ corresponds to contracting a necessary redex in order to eventually substitute $v$ by its definition.\bigskip

Consider for instance the expression $e$ given in Example \ref{ex:inlinevar-contractlam-makeredex}. To go towards the inlining of the variable $tl$ in the configuration of the $\oscan$ operator, the heuristic calls $\hinlinevar{e}{tl}$, which tries to contract the binding lambda by calling $\hcontractlam{e}{\elam tl.\ \oscan_{f\ tl}()}$. This lambda cannot be contracted before the lambda above it, and therefore $\hcontractlam{e}{\elam hd.\ \elam tl.\ \ldots}$ is called recursively. Again, this lambda cannot be contracted before the head $\eldestr$, but this requires its first child expression to be reduced to a $\elcons$. This is the reason why the heuristic recursively calls $\hmakeredex{e}{(\elam x.\ x)\ (\elcons\ \ldots)}$, which contracts its argument redex.

\begin{example}
\label{ex:inlinevar-contractlam-makeredex}
\onestepreductionexample{
	$\eldestr\ ((\elam x.\ x)\ (\elcons\ 1\ (\elcons\ 2\ \elnil)))$\\
	$\efalse\ (\elam hd.\ \elam tl.\ \oscan_{f\ tl}())$
}{
  $\eldestr\ (\elcons\ 1\ (\elcons\ 2\ \elnil))$\\
	$\efalse\ (\elam hd.\ \elam tl.\ \oscan_{f\ tl}())$
}
\end{example}

\begin{definition}
\label{def:cfgonered}
Let $e$ be an expression, $C[]$ an operator context in $\oppaths{e}$, $V = \listcomp{v \in \cfgfreevars{e}{C[]}}{\hinlinevar{e}{v} \neq \hnone}$ and $R = \cfgredexes{e}{C[]}$. We define the one-step reduction of the operator configuration $\honeredcfg{e}{C[]}$ as
\begin{itemize}
	\item $\hinlinevar{e}{v}$ such that $v$ is the first element of $L$, if $L$ is non empty
	\item $e'$, such that $e \reduceredex{r} e'$ and $r$ is the first element of $R$, if $R$ is non empty
	\item $\hnone$ otherwise
\end{itemize}
\end{definition}

This means that for a given operator, the heuristic first tries to find a free variable that can be inlined, then if there is no such variable, it reduces the configuration using a leftmost outermost (call-by-name) reduction strategy, and finally if there is no redex to reduce, it returns $\hnone$.\bigskip

In the following definitions, we will describe (search space) trees of expressions using the notations given in Section \ref{sec:notations}.

\begin{definition}
\label{def:hredcfg}
Let $e$ be an expression, $C[]$ an operator context in $\oppaths{e}$, $\phi$ an integer and $e' = \honeredcfg{e}{C[]}$. We define $\hredcfg{\phi}{e}{C[]}$ as
\begin{itemize}
	\item $e()$ if $\phi = 0$ or if $e' = \hnone$
	\item $e'(\listcomp{\hredcfg{\Phi}{e'}{op}}{op \in \oppaths{e'}})$ if $M(e') < M(e)$
	\item $e(\listcomp{\hredcfg{\phi-1}{e'}{op}}{op \in \oppaths{e'}})$ otherwise
\end{itemize}
\end{definition}

\begin{definition}
\label{def:pass1}
Let $e$ be an expression. The search space $\hconfigss{\Phi}{e}$ after the first pass is defined as $e(\listcomp{\hredcfg{\Phi}{e}{op}}{op \in \oppaths{e}})$.
\end{definition}

$\hconfigss{\Phi}{e}$ corresponds to the recursive exploration of the search space tree for reductions in configurations, with the guarantee that search space subtrees in which the measure does not decrease have a depth bounded by $\Phi$.

\begin{definition}
\label{def:pass1result}
Let $e$ be an expression. The result expression $\hconfig{\Phi}{e}$ of the first pass is defined as follows. First, the rewrite rule  $x(s_1,\ldots,s_i,x(),s_{i+2},s_n) \rightarrow x(s_1,\ldots,s_i,s_{i+2},s_n)$ is applied on $\hconfigss{\Phi}{e}$ as much as possible. Then, the leftmost leaf is chosen.
\end{definition}

The rewrite rule corresponds to pruning search space subtrees in $\hconfigss{\Phi}{e}$ that failed to reach an expression with smaller measure, whereas taking the leftmost leaf is a heuristic decision corresponding to never backtrack to expressions with larger measure than the current candidate (for performance reasons).

\begin{definition}
\label{def:childonered}
Let $e$ be an expression, $C[]$ an operator context in $\oppaths{e}$ and $L = \listcomp{c \in \opchildren{e}{C[]}}{\hmakeredex{e}{c} \neq \hnone}$. We define the one-step reduction of the operator children $\honeredchild{e}{C[]}$ as
\begin{itemize}
	\item $\hmakeredex{e}{c}$ such that $c$ is the first element of $L$, if $L$ is non empty
	\item $\hnone$ otherwise
\end{itemize}
\end{definition}

This means that for a given operator, the heuristic tries to find a child expression that can be reduced (using a call-by-name reduction strategy) in order to regroup fragments, and if there is no such child, it returns $\hnone$.

\begin{definition}
\label{def:hredchild}
Let $e$ be an expression, $C[]$ an operator context in $\oppaths{e}$, $\phi$ an integer and $e' =\honeredchild{e}{C[]}$. We define $\hredchild{\phi}{e}{C[]}$ as
\begin{itemize}
	\item $e()$ if $\phi = 0$ or if $e' = \hnone$
	\item $e'(\listcomp{\hredchild{\Phi}{e'}{op}}{op \in \oppaths{e'}})$ if $M(e') < M(e)$
	\item $e(\listcomp{\hredchild{\phi-1}{e'}{op}}{op \in \oppaths{e'}})$ otherwise
\end{itemize}
\end{definition}

\begin{definition}
\label{def:pass2}
Let $e$ be an expression. The search space $\hchildss{\Phi}{e}$ after the second pass is defined as $e(\listcomp{\hredchild{\Phi}{e}{op}}{op \in \oppaths{e}})$.
\end{definition}

$\hchildss{\Phi}{e}$ corresponds to the recursive exploration of the search space tree for reduction in children, with the guarantee that search space subtrees in which the measure does not decrease have a depth bounded by $\Phi$.

\begin{definition}
\label{def:pass2result}
Let $e$ be an expression. The result expression $\hchild{\Phi}{e}$ of the second pass is defined as follows. First, the rewrite rule  $x(s_1,\ldots,s_i,x(),s_{i+2},s_n) \rightarrow x(s_1,\ldots,s_i,s_{i+2},s_n)$ is applied on $\hchildss{\Phi}{e}$ as much as possible. Then, the leftmost leaf is chosen.
\end{definition}

Similarly to Definition \ref{def:pass1result}, the rewrite rule corresponds to pruning search space subtrees in $\hchildss{\Phi}{e}$ that failed to reach an expression with smaller measure, whereas taking the leftmost leaf is a heuristic decision.

\begin{definition}
\label{def:hminred}
Let $e$ be an expression. The result expression $\hminred{\Phi}{e}$ of the heuristic reduction is defined as $\hchild{\Phi}{\hconfig{\Phi}{e}}$.
\end{definition}

Figure \ref{fig:heuristic} illustrates the construction of the two search spaces (cf. Definitions \ref{def:pass1} and \ref{def:pass2}) leading to the computation of $\hminred{\Phi}{e}$. Dots and stars represent expressions considered by the heuristic as a possible reduction of their parent, but stars correspond to the special case where the measure of the expression is smaller than the measure of its parent, i.e., when the heuristic made progress.\bigskip

We will now continue towards a proof of termination of this reduction strategy (Theorem \ref{th:termination}).

\begin{definition}
\label{def:traversal-pos}
Considering an expression with subexpressions as a tree with subtrees, we define a traversal of an expression as follows, starting from the head expression: 
\begin{itemize}
	\item for a nullary expression, visit this expression.
	\item for $\elam x.\ e_1$, $\elcons\ e_1\ e_2$ and $\etcons\ s\ e_1\ e_2$, visit this expression then traverse the children from left to right ($e_1$ then $e_2$ $\ldots$ then $e_n$).
	\item for other $n$-ary expressions (e.g., $(e_1\ e_2)$, $\eldestr\ e_1\ e_2\ e_3$, etc.), traverse the children from left to right ($e_1$ then $e_2$ $\ldots$ then $e_n$), then visit the expression.
\end{itemize}
We denote by $\travpos{e}{e'}$ the position of a subexpression $e'$ in this traversal of $e$.
\end{definition}

$\travpos{e}{.}$ actually induces a well-founded order on the locations of the subexpressions that are arguments of the recursive calls to $\hinlinevar{.}{.}$, $\hmakeredex{.}{.}$ and $\hcontractlam{.}{.}$ made by the heuristic. It will be used in inductions in some of the following proofs.

\noindent\begin{minipage}[l]{0.52\linewidth}
\centering
\captionsetup{type=figure}
\includegraphics[width=0.825\linewidth]{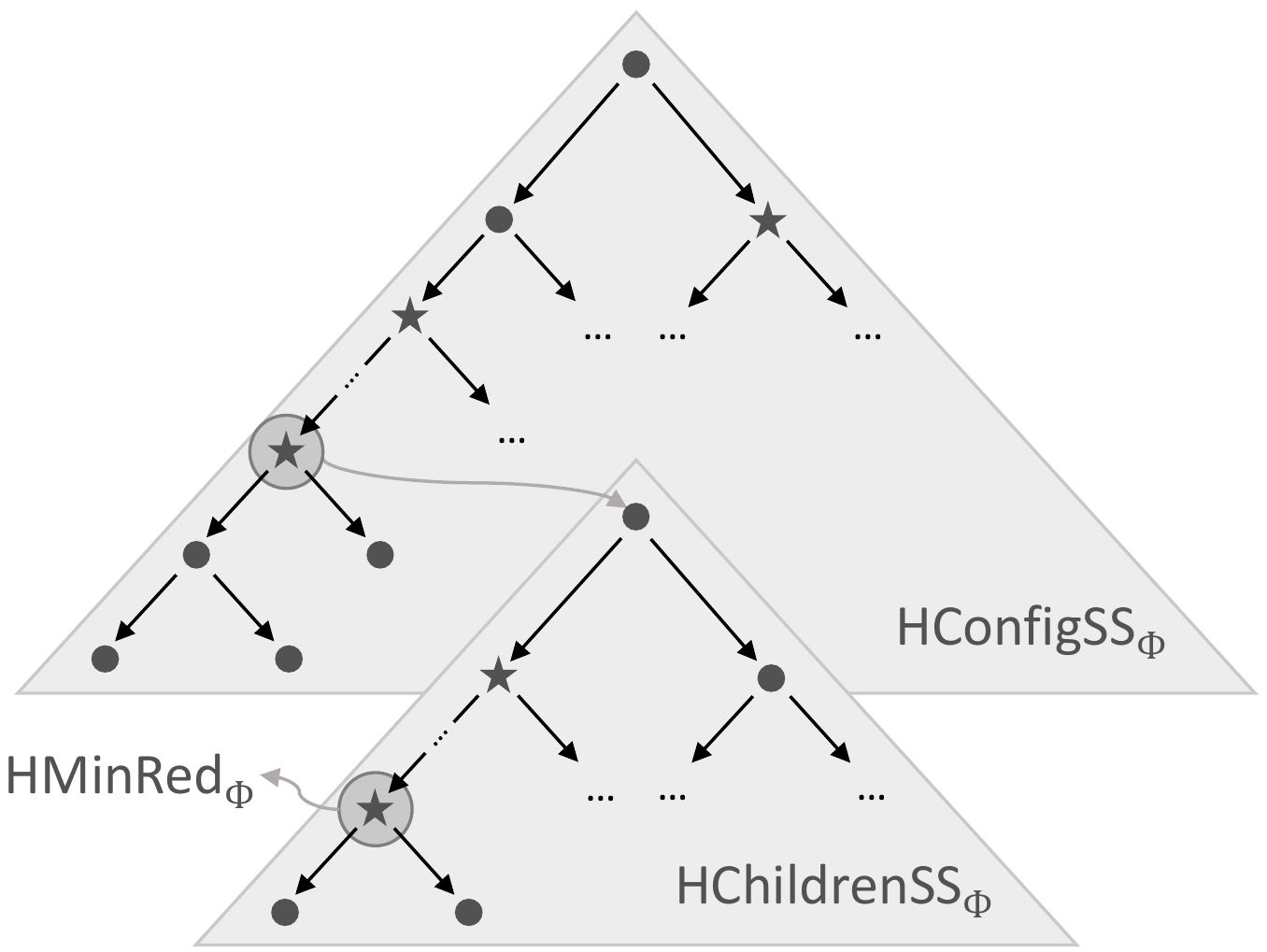}
\captionof{figure}{$\hminred{\Phi}{e}$ and the search spaces $\hconfigss{\Phi}{e}$ and $\hchildss{\Phi}{\hconfig{\Phi}{e}}$}
\label{fig:heuristic}
\end{minipage}%
\hspace{0.04\linewidth}%
\begin{minipage}[c]{0.44\textwidth}
\centering
\captionsetup{type=figure}
\includegraphics[width=0.8\linewidth]{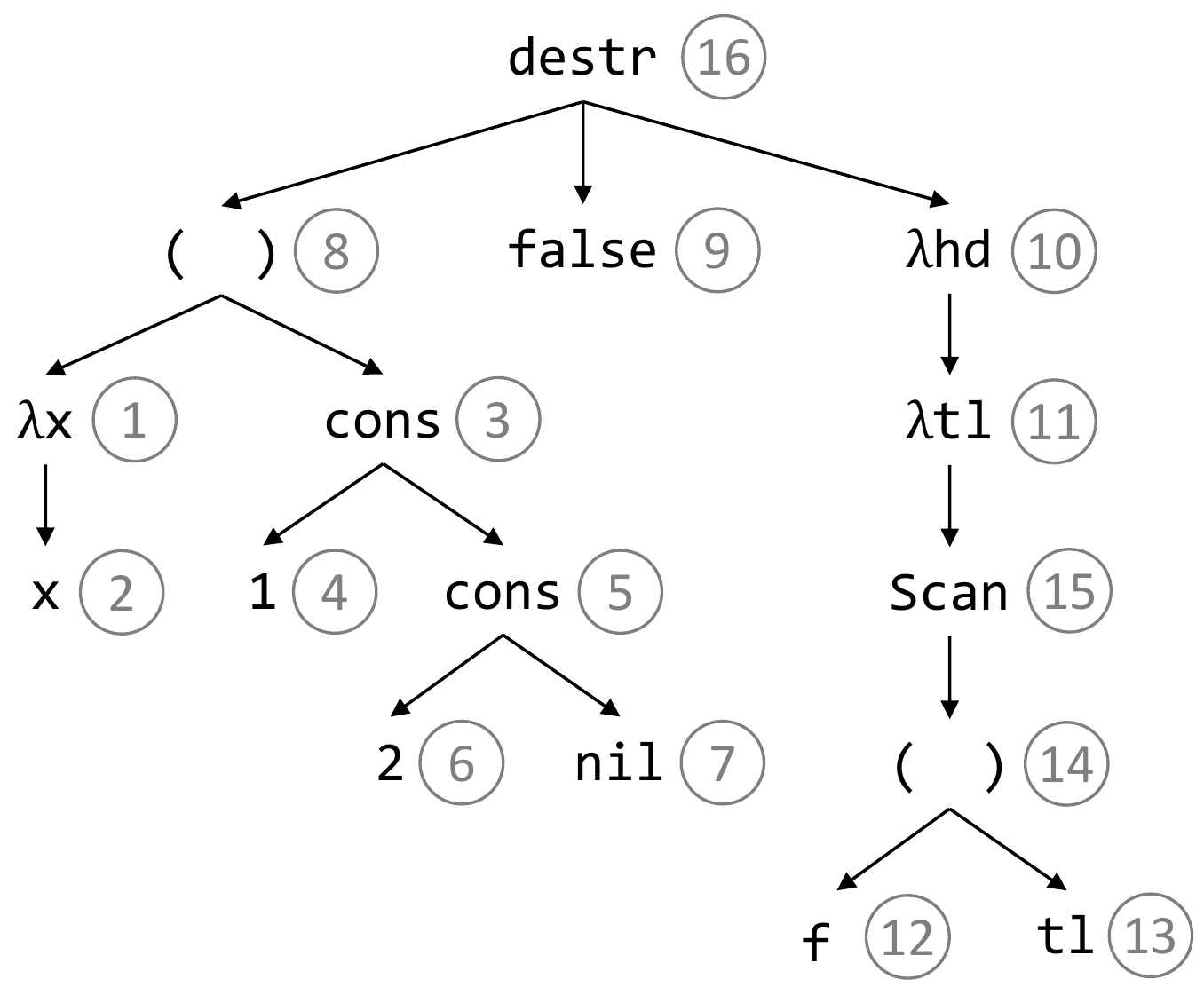}
\captionof{figure}{Example of expression traversal}
\label{fig:traversal}
\end{minipage}

\vspace{1em}

For example, consider again the expression $e$ from Example \ref{ex:inlinevar-contractlam-makeredex} and its traversal given in Figure \ref{fig:traversal}. Remember the chain of recursive calls made by the heuristic ($\hinlinevar{e}{tl}$, which calls $\hcontractlam{e}{\elam tl.\ \oscan_{f\ tl}()}$, which calls $\hcontractlam{e}{\elam hd.\ \elam tl.\ \ldots}$, which indirectly calls $\hmakeredex{e}{(\elam x.\ x)\ (\elcons\ \ldots)}$, which finally contracts this redex) and notice that $\travpos{e}{tl} > \travpos{e}{\elam tl.\ \oscan_{f\ tl}()} > \travpos{e}{\elam hd.\ \elam tl.\ \ldots} > \travpos{e}{(\elam x.\ x)\ (\elcons\ \ldots)}$.

\begin{restatable}{lemma}{lemhinlinevarhmakeredexhcontractlamterminate}
\label{lem:hinlinevar-hmakeredex-hcontractlam-terminate}
For all expression $e$ the following properties hold.
\begin{itemize}
	\item For a variable $e'$ in $e$, $\hinlinevar{e}{e'}$ returns $\hnone$ or contracts a redex.
	\item For a subexpression $e'$ in $e$, $\hmakeredex{e}{e'}$ returns $\hnone$ or contracts a redex.
	\item For a subexpression $e'$ in $e$, $\hcontractlam{e}{e'}$ returns $\hnone$ or contracts a redex.
\end{itemize}
\end{restatable}

These three properties not only show that the objects used to describe the heuristic are well defined, but also mean that correct implementations of $\hinlinevar{.}{.}$, $\hmakeredex{.}{.}$ and $\hcontractlam{.}{.}$ terminate.

\begin{restatable}{lemma}{lemhredcfghredchildterminate}
\label{lem:hredcfg-hredchild-terminate}
For all expression $e$, the search space $\hconfigss{\Phi}{e}$ (resp. $\hchildss{\Phi}{e}$) has bounded size.
\end{restatable}

\begin{theorem}
\label{th:termination}
This reduction strategy always terminates, that is, for an input expression $e$, it always considers a finite number of expressions $e'$ such that $e \rightarrow^* e'$ in order to find $\hminred{\Phi}{e}$.
\begin{proof}
Using Lemma \ref{lem:hredcfg-hredchild-terminate}, this reduction strategy only considers a finite number of expressions $e_1$ such that $e \rightarrow^* e_1$ in order to find $\hconfig{\Phi}{e}$, then using Lemma \ref{lem:hredcfg-hredchild-terminate} again, it only considers a finite number of expressions $e_2$ such that $\hconfig{\Phi}{e} \rightarrow^* e_2$ in order to find $\hchild{\Phi}{\hconfig{\Phi}{e}} = \hminred{\Phi}{e}$.
\end{proof}
\end{theorem}

We will now continue towards a proof that under some hypotheses this reduction strategy is complete, that is, it returns an optimal result for a big enough fuel value (Theorem \ref{th:completeness}).\bigskip

In the following proofs, we will need a way to track subexpressions of an expression and understand how they are duplicated, erased or simply moved across successive reductions of the main expression. Borrowing an idea used in \cite{lambda-calc} to define residuals, we annotate all subexpressions of the original expression with unique identifiers, simply propagate such identifiers along reductions without altering them, and reason on identifiers in the final expression.

\begin{definition}
\label{def:occurences}
Let $e_0$ be a QIR expression, $e$ a subexpression of $e_0$ and $e_n$ a QIR expression such that $e_0 \rightarrow^* e_n$. We define $\occur{e_n}{e}$ as the set of subexpressions of $e_n$ with the same identifiers as $e$ when identifiers are propagated along the reduction chain $e_0 \rightarrow^* e_n$.
\end{definition} 

For instance, the QIR reduction $(\elam x.\ \elam y.\ f\ x\ x)\ z\ z \rightarrow (\elam y.\ f\ z\ z)\ z \rightarrow f\ z\ z$\\
can be annotated as follows: $(\elam^1x.\ \elam^2y.\ f^3\ x^4\ x^5)\ z^6\ z^7 \rightarrow (\elam^2y.\ f^3\ z^6\ z^6)\ z^7 \rightarrow f^3\ z^6\ z^6$.\\
\noindent One can see, just by looking at identifiers, that the first $z$ variable produced both $z$ variables of the final expression: it has two occurences in the final expression. The other one disappeared: it has no occurence in the final expression.\bigskip

A more formal definition of $\occur{e_n}{e}$ is given in Appendix \ref{ap:qir-star}. The following definitions use the $\occur{.}{.}$ notation but do not depend on the internals of its definition.

\begin{definition}
\label{def:erased}
Let $e$ be an expression, $e'$ a subexpression of $e$ and $r$ a redex such that $e \reduceredex{r} e''$ and $\occur{e''}{e'} = \emptyset$. Borrowing a concept from \cite{lambda-calc}, we say that $r$ \emph{erases} $e'$ if $e'$ is not part of the redex $r$, that is if $e'$ is a subexpression of an expression $e''$ such that one of the following conditions holds.
\vspace{-1em}
\begin{multicols}{2}
\begin{itemize}
	\item $r = (\elam x.\ e_1)\ e''$ and $x$ is not free in $e_1$
	\item $r = \eldestr\ \elnil\ e_1\ e''$
	\item $r = \eldestr\ (\elcons\ \ldots)\ e''\ e_1$
	\item $r = \etdestr \ (\etcons \ \estr{s1} \ e_1 \ e'') \ \estr{s1}$
	\item $r = \etdestr \ (\etcons \ \estr{s1} \ e'' \ e_1) \ \estr{s2}$
	\item $r = \eif\ \etrue\ \ethen\ e_1\ \eelse\ e''$
	\item $r = \eif\ \efalse\ \ethen\ e''\ \eelse\ e_1$
	\item $r = true\ \eand\ e''$ (and similarly for other $\rho$-redexes)
\end{itemize}
\end{multicols}
\end{definition}

\begin{definition}
\label{def:consumed}
Let $e$ be an expression, $e'$ a subexpression of $e$ and $r$ a redex such that $e \reduceredex{r} e''$ and $\occur{e''}{e'} = \emptyset$. We say that $r$ \emph{consumes} $e'$ if $e'$ is part of the redex $r$, that is if one of the following conditions holds.
\vspace{-1em}
\begin{multicols}{2}
\begin{itemize}
	\item $r = e'$
	\item $r = (e'\ e_1)$
	\item $r = (\elam v.\ e_1)\ e_2$ and $e'$ is the variable $v$ free in $e_1$
	\item $r = \eldestr\ e'\ e_1\ e_2$
	\item $r = \etdestr\ e'\ s$
	\item $r = \eif\ e'\ \ethen\ e_1\ \eelse\ e_2$
	\item $r = e'\ \eand\ e_1$ (and similarly for other $\rho$-redexes)
\end{itemize}
\end{multicols}
\end{definition}
\vspace{-0.8em}

\begin{restatable}{lemma}{lemerasedorconsumed}
\label{lem:erased-or-consumed}
Let $e$ be an expression, $e'$ a subexpression of $e$ and $r$ a redex in $e$. If $e \reduceredex{r} e''$ and $\occur{e''}{e'} = \emptyset$, then $r$ either consumes or erases $e'$.
\end{restatable}

\begin{definition}
\label{def:duplicate}
Let $e$ be an expression, $e'$ a subexpression of $e$ and $r$ a redex in $e$ such that $e \reduceredex{r} e''$, we say that $r$ \emph{duplicates} $e'$ if $|\occur{e''}{e'}| > 1$.
\end{definition}

\begin{definition}
\label{def:not-erased-consumed}
Let $e_0$ be an expression and $e'$ a subexpression of $e_0$. We say that $e'$ \emph{cannot be erased} (resp. \emph{consumed}, \emph{duplicated}) if there is no reduction chain $e_0 \reduceredex{r_1} \ldots \reduceredex{r_n} e_n$ such that $r_n$ erases (resp. consumes, duplicates) an element of $\occur{e_{n-1}}{e'}$.
\end{definition}
\vspace{-0.8em}

\begin{restatable}{lemma}{leminlinevarmakeredexcontractlamcomplete}
\label{lem:inlinevar-makeredex-contractlam-complete}
Let $e$ be an expression and $e'$ a subexpression of $e$. The following properties hold.
\begin{itemize}
	\item If $e'$ is a variable $v$ and if $\hinlinevar{e}{v}$ returns $\hnone$ then $e'$ cannot be consumed.
	\item If $\hmakeredex{e}{e'}$ returns $\hnone$ then $e'$ cannot be consumed.
	\item If $\hcontractlam{e}{e'}$ returns $\hnone$ then $e'$ cannot be consumed.
\end{itemize}
\end{restatable}

\begin{definition}
\label{def:fixed-operators}
We say that an expression $e$ has the fixed operators property if no operator in $e$ can be erased or duplicated.
\end{definition}

Remember that the definition of compatible operators (Definitions \ref{def:supported-op}, \ref{def:supported-expr} and \ref{def:compatible-op}) depends on an arbitrary external module describing the capabilities of the target database (cf. Section \ref{sec:architecture}). The following restriction allows us to reason on this module.

\begin{definition}
\label{def:compat-stable}
We say that a database verifies the stable compatibility property if, given an expression $e$, an operator $op$ in $\{ op\ |\ C[] \in \oppaths{e},\ e = C[op] \}$ such that $op$ is compatible and an expression $e'$ such that $e \rightarrow^* e'$, each operator $op' \in \occur{op}{e'}$ is also compatible.
\end{definition}

This last definition should hold for a realistic database and an accurate description of its capabilities. Indeed, it basically says that if an operator is compatible, any reduction either does not affect the operator or helps the database by simplifying its configuration.

\begin{restatable}{lemma}{lemreducedecrmeasure}
\label{lem:reduce-decr-measure}
Let $e$ be an expression with fixed operators and $r$ a redex in $e$. For a database with stable compatibility, if $e \reduceredex{r} e'$ then $M(e') \leq M(e)$.
\end{restatable}
\vspace{-0.8em}

\begin{restatable}{lemma}{lemnormalmeasuremin}
\label{lem:normal-measure-min}
Let $e$ be a weakly-normalizing expression with fixed operators. For a database with stable compatibility, the normal form of $e$ has minimal measure.
\end{restatable}
\vspace{-0.8em}

\begin{restatable}{lemma}{lemsamecompatoperators}
\label{lem:same-compat-operators}
Let $e$ be a weakly-normalizing expression with fixed operators, $e_{min}$ an expression in $\minreductions{e}$ and $e'$ an expression such that $e \rightarrow^* e'$ and $\op{e'} - \comp{e'} = \op{e_{min}} - \comp{e_{min}}$. For a database with stable compatibility, an operator is compatible in $e_{min}$ if and only if it is compatible in $e'$.
\end{restatable}

\begin{theorem}
\label{th:completeness}
For databases with stable compatibility, the reduction strategy of Definition \ref{def:hminred} is complete on strongly-normalizing expressions with fixed operators. That is, for a database with stable compatibility, given a strongly-normalizing input expression $e$ with fixed operators, there exists a fuel value $\Phi$ such that $\hminred{\Phi}{e} \in \minreductions{e}$.
\begin{proof}
Remember from Definition \ref{def:exhaustive-red} that all expressions in $\minreductions{e}$ have same (minimal) measure. Using Lemma \ref{lem:normal-measure-min}, we know that the normal form $e_N$ of $e$ is in $\minreductions{e}$. Let $M_{min}$ be its measure. Consider now $e_h = \hminred{\Phi}{e}$. Using Theorem \ref{th:optimal}, we know that $e_h \in \minreductions{e}$ or $M(e_h) > M_{min}$ and we want to prove that the latter cannot happen. Suppose to the contrary that $M(e_h) > M_{min}$. Using the definition of $M$ (cf. Definition \ref{def:measure}), this means that one of the two following statements holds.
\begin{itemize}
	\item $\op{e_h} - \comp{e_h}$ is greater than the first component of $M_{min}$
	\item $\op{e_h} - \comp{e_h}$ is equal to the first component of $M_{min}$ and $\frag{e_h}$ is greater than the second component of $M_{min}$
\end{itemize}
We will prove that none of these cases can happen.
\begin{itemize}
	\item Suppose that $\op{e_h} - \comp{e_h}$ is greater than the first component of $M_{min}$. Since $e$ has the fixed operators property, there is a one-to-one correspondance between the operators of $e_N$ and $e_h$. Therefore, we know that $\comp{e_h} < \comp{e_N}$ and there exists an operator $op$ in $e$ such that $\occur{e_h}{op} = \{ op_h \}$, $\occur{e_N}{op} = \{ op_N \}$, $op_N$ is compatible and $op_h$ is not compatible. Let $c_h$ (resp. $c_N$) be the configuration of $op_h$ (resp. $op_N$). The question now is to understand how the first pass of the heuristic-based algorithm (cf. Definition \ref{def:pass1result}) could fail to make $op$ compatible. Remember Lemma \ref{lem:hinlinevar-hmakeredex-hcontractlam-terminate} telling that $\hinlinevar{.}{.}$ and $\hmakeredex{.}{.}$ either contract a redex or return $\hnone$, and keep in mind that such reductions maintain a single instance of $c_h$ in the reduced forms of $e_h$ (fixed operator hypothesis). Since $e$ is strongly-normalizing, this means that there is a fuel value $\Phi$ allowing the heuristic to make enough calls to $\hinlinevar{.}{.}$ on the free variables of $c_h$ in order to get to an expression $e_h' = C'[op_{c_h'}(\ldots)]$ such that (i) there is no free variable in $c_h'$ or (ii) calls to $\hinlinevar{e_h'}{.}$ return $\hnone$ for every free variable of $c_h'$. Continuing from this point, since $e$ is strongly-normalizing, $e_h'$ and $c_h'$ are also strongly normalizing. Thus, Theorem \ref{th:standardization} tells that there is a fuel value $\Phi$ allowing the heuristic to reduce all redexes of $c_h'$ and reach a normal form $c_h''$ and an expression $e_h'' = C'[op_{c_h''}(\ldots)]$. Since we supposed that the heuristic failed to make $op$ compatible, this means that $c_h''$ is different from $c_N$. Using Theorem \ref{th:confluence} (confluence), we know there is reduction $e_h'' \rightarrow^* e_N$. Since the redexes contracted in this chain cannot erase nor duplicate operators (fixed operator hypothesis), the reduction chain can only affect $c_h''$ in the following ways.
	\begin{itemize}
		\item Substitute free variables in $c_h''$. This cannot happen: by hypothesis, either (i) there is no free variable in $c_h'$ and therefore in $c_h''$ or (ii) calls to $\hinlinevar{e_h'}{.}$ return $\hnone$ for every free variable of $c_h'$ and using Lemma \ref{lem:inlinevar-makeredex-contractlam-complete}, such a variable cannot be consumed.
		\item Reduce redexes located inside $c_h''$. This cannot happen since $c_h''$ is in normal form.
		\item Leave $c_h''$ untouched. This leads to a contradiction: $c_h''$ is equal to $c_N$.
	\end{itemize}
Therefore, there is a fuel value such that the heuristic makes $op$ compatible. Now, taking the maximum of the required values of $\Phi$ to make each operator compatible, there exists a value of $\Phi$ such that $\op{e_h} - \comp{e_h}$ is equal to the first component of $M_{min}$.
	\item Suppose now that $\op{e_h} - \comp{e_h}$ is equal to the first component of $M_{min}$ and $\frag{e_h}$ is greater than the second component of $M_{min}$. Since $e$ has the fixed operators property, there is a one-to-one correspondance between the operators of $e_N$ and $e_h$. Using Lemma \ref{lem:same-compat-operators}, we know that there exists an operator $op$ in $e$ such that $\occur{e_h}{op} = \{ op_h \}$, $\occur{e_N}{op} = \{ op_N \}$, $op_N$ and $op_h$ are both compatible, $op_N$ has a compatible child operator $c_N$ and the child expression $c_h$ of $op_h$ is incompatible (i.e., not a compatible operator). The question now is to understand how the second pass of the heuristic-based algorithm (cf. Definition \ref{def:pass2result}) could fail to reduce $c_h$ to a compatible operator. Remember Lemma \ref{lem:hinlinevar-hmakeredex-hcontractlam-terminate} telling that $\hmakeredex{.}{.}$ either contracts a redex or returns $\hnone$, and keep in mind that such reductions maintain a single instance of $op_h$ in the reduced forms of $e_h$ (fixed operator hypothesis). Since $e$ is strongly-normalizing, this means that there is a fuel value $\Phi$ allowing the heuristic to make enough calls to $\hmakeredex{.}{.}$ on $c_h$ in order to get to an expression $e_h' = C'[op_{\ldots}(\ldots, c_h', \ldots)]$ such that calls to $\hmakeredex{c_h'}{.}$ returns $\hnone$. Since we supposed that the heuristic failed to reduce $c_h$ to a compatible operator, this means that the head of $c_h'$ is different from the head of $c_N$ (which is a compatible operator). Using Lemma \ref{lem:inlinevar-makeredex-contractlam-complete}, $c_h'$ cannot be consumed, and as the child expression of an operator that cannot be erased, $c_h'$ cannot be erased either. According to Lemma \ref{lem:erased-or-consumed} this contradicts the confluence theorem telling that $e_h' \rightarrow^* e_N$. Therefore, there is a fuel value such that the heuristic reduces $c_h$ to a compatible operator. Now, taking the maximum of the required values of $\Phi$ to reduce the children of all operators, there exists a value of $\Phi$ such that $\frag{e_h}$ is equal to the second component of $M_{min}$.
\end{itemize}
\vspace{-1.2em}
\end{proof}
\end{theorem}

We also conjecture that the result of Theorem \ref{th:completeness} still holds for weakly-normalizing expressions.

\begin{restatable}{conjecture}{conjcompleteness}
\label{conj:completeness}
For databases with stable compatibility, the reduction strategy of Definition \ref{def:hminred} is complete on weakly-normalizing expressions with fixed operators. That is, for a database with stable compatibility, given a weakly-normalizing input expression $e$ with fixed operators, there exists a fuel value $\Phi$ such that $\hminred{\Phi}{e} \in \minreductions{e}$.
\end{restatable}

See Appendix \ref{ap:proofs} for an intuition of why this holds and the current state of the proof. We will now explain why none of the remaining hypotheses can be removed.
\vspace{-0.7em}

\paragraph{Stable compatiblity} Consider a database for which the $\oscan$ operator is compatible if its configuration has more than two free variables. Obviously, it would not have the stable compatibility property, since inlining the definition of these variables could reduce the number of free variables in the configuration. Take now expression $e$ from Example \ref{ex:stable-compat}. Since the heuristic tries to inline variables before reducing redexes in the configurations, it will never consider expression $e'$, which is the only element of $\minreductions{e}$.
	
\begin{example}
\label{ex:stable-compat}
\multistepreductionexample{
  $(\elam t.\ \oscan_{(\elam x.\ x = x)\ t}())\ 1$
}{
  $(\elam t.\ \oscan_{\elam x.\ t = t}())\ 1$
}
\end{example}
	
For the next counterexamples, we will suppose a simplistic database capabilities description, for which all operators are compatible as long as there is no free variable in their configuration (such a database would have the stable compatibility property).
\vspace{-0.7em}

\paragraph{Non-normalizing expressions} Consider expression $e$ from Example \ref{ex:normalizing}. Obviously, it is non-normalizing because of the $\Omega = (\elam x.\ x\ x)\ (\elam x.\ x\ x)$ in the configuration. Since the heuristic applies a call-by-name reduction strategy on the configurations once all free variables are inlined, it will consume all the fuel on $\Omega$ and never consider $e'$, which is the only element of $\minreductions{e}$.

\begin{example}
\label{ex:normalizing}
\multistepreductionexample{
  $\oscan_{((\elam x.\ x\ x)\ (\elam x.\ x\ x))\ ((\elam x.\ 1)\ y)}()$
}{
  $\oscan_{((\elam x.\ x\ x)\ (\elam x.\ x\ x))\ 1}()$
}
\end{example}

\paragraph{Operator erasure} Consider expression $e$ from Example \ref{ex:erase-op}. Obviously, the $\oscan_x()$ operator can be erased. Since the heuristic tries to inline variables in configurations, reduce configurations then regroup fragments, it will never consider $e'$, which is the only element of $\minreductions{e}$.
	
\begin{example}
\label{ex:erase-op}
\multistepreductionexample{
  $\eif\ \efalse\ \ethen\ \oscan_{x}()\ \eelse\ \oscan_{\edataref{table}}()$
}{
  $\oscan_{\edataref{table}}()$
}
\end{example}

\paragraph{Operator duplication} Consider expression $e$ from Example \ref{ex:duplicate-op}. Obviously, the $\oscan_{z}$ operator can be duplicated. The heuristic will try to inline $y$ in the configuration of the two $\ojoin$ operators, which requires to inline $x$ first. Since this two-step reduction decreases the measure and because the heuristic chooses the leftmost leaf of the configuration search space, $e'$ will never be considered although it is the only element of $\minreductions{e}$.

\begin{example}
\label{ex:duplicate-op}
\multistepreductionexample{
  $(\elam x.\ \elam y.\ \ojoin_{\eif\ \efalse\ \ethen\ y\ \eelse\ \etrue}($\\
	\hspace*{0.5em}$\oscan_{\edataref{table}}(),\ojoin_{\eif\ \efalse\ \ethen\ y\ \eelse\ \etrue}(x,x)$\\
	$))\ \oscan_z()\ \efalse$
}{
  $(\elam x.\ \elam y.\ \ojoin_{\etrue}($\\
	\hspace*{0.5em}$\oscan_{\edataref{table}}(),\ojoin_{\etrue}(x,x)$\\
	$))\ \oscan_z()\ \efalse$
}
\end{example}

\noindent\begin{minipage}[l]{0.38\linewidth}
\centering
\captionsetup{type=figure}
\includegraphics[width=\linewidth]{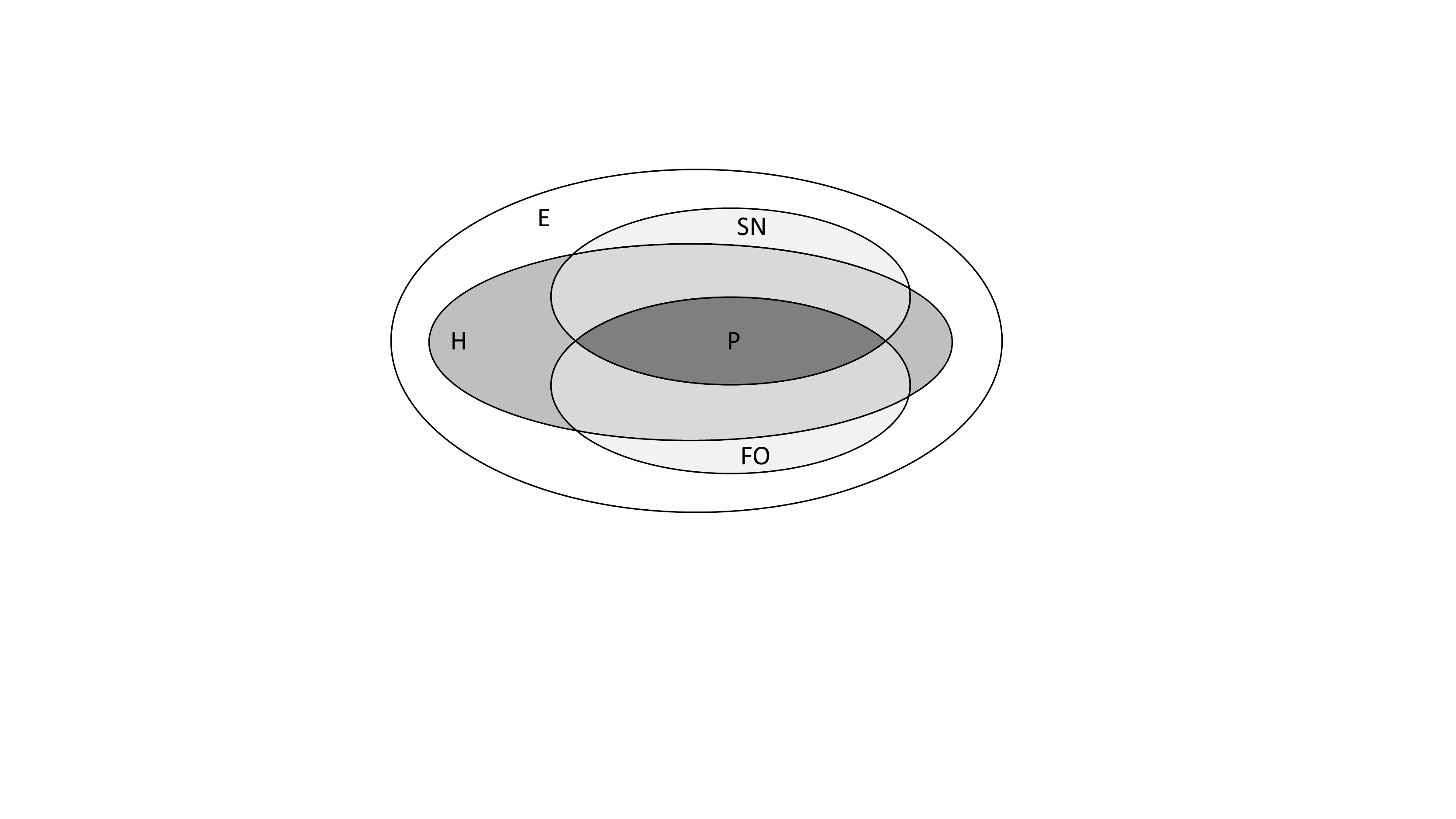}
\captionof{figure}{Situation summary}
\label{fig:venn}
\end{minipage}%
\hspace{0.02\linewidth}%
\begin{minipage}[c]{0.6\textwidth}
Figure \ref{fig:venn} sums up the situation. In this Venn diagram, $E$ stands for the set of all QIR expressions, $SN$ for the set of strongly-normalizing QIR expressions, $FO$ for the set of QIR expression with fixed operators, $H$ for the set of QIR expressions on which the heuristic returns an optimal result and $P$ for the set of QIR expressions for which I proved that the heuristic returns such a result. As I will show in Section \ref{sec:concrete-examples}, $H$ is larger than $P$: in fact the heuristic returns optimal results on all our real-world use cases.
\end{minipage}

\vspace{1em}
In Appendix \ref{ap:fuel}, I discuss how the fuel value $\Phi$ can be estimated by analyzing the QIR code. Though, an implementation of a module doing such an estimation is future work, and for now, $\Phi$ has to be set by the developer.
\section{Experiments}
\label{sec:concrete-examples}

In this section, I present real-world use cases and show how the QIR evaluation strategies discussed in Section \ref{sec:qir-evaluation} perform on these examples. In the following experiments, the target database language is SQL and for conciseness, I never show the application code but rather quickly describe the applications themselves.\bigskip

The provided timings have been measured on my personal laptop (Core i5 M450 @2.4GHz, 8Go RAM). Although my setup is not precise enough for fine-grained performance comparison, it enables an order of magnitude analysis and delivers interesting results.

\subsection{Example of code factorization: analytics queries}
\label{sec:ex:analytics}

Consider the following analytics query written in SQL and freely inspired by the \emph{TPC-H} benchmark.

\begin{lstlisting}[language=sql,basicstyle=\small]
SELECT
  l_returnflag AS return_flag, l_linestatus AS line_status,
  SUM(l_extended_price) AS sum_base_price,
  SUM(l_extended_price * (1 - l_discount)) AS sum_disc_price,
  SUM(l_extended_price * (1 - l_discount) * (1 + l_tax)) AS sum_charge,
  SUM(l_extended_price) * 0.75 AS sum_real_cost,
  SUM(l_extended_price) * 0.25 AS sum_margin,
  AVG(l_extended_price) AS avg_base_price,
  AVG(l_extended_price * (1 - l_discount)) AS avg_disc_price,
  AVG(l_extended_price * (1 - l_discount) * (1 + l_tax)) AS avg_charge,
  AVG(l_extended_price) * 0.75 AS avg_real_cost,
  AVG(l_extended_price) * 0.25 AS avg_margin
FROM db.lineitem
GROUP BY l_return_flag, l_linestatus 
ORDER BY l_return_flag, l_linestatus
\end{lstlisting}

Notice in particular the many common expressions that are used to compose the aggregation functions of the query. To factorize code and increase the maintainability of the code, the application developer might want to define these expressions only once and reuse them throughout the query. For better readability, he might also store subexpressions in variables instead of writing the query in one block. The direct translation of this kind of code to QIR would lead to the following query representation (remember that in QIR, as opposed to SQL's strange syntax, aggregations functions belong to the configuration of the $\ogroup$ operator).

\begin{lstlisting}[language=C,basicstyle=\small]
/* Constructs the list of projected attributes */
$\elet\ project\_list = \elam tup.$
  $\elet\ real\_cost = 0.75\ \ein$
  $\elet\ margin = 0.25\ \ein$
  $\elet\ return\_flag = \etdestr\ tup\ \estr{l_returnflag}\ \ein$
  $\elet\ line\_status = \etdestr\ tup\ \estr{l_linestatus}\ \ein$
  $\elet\ sum\_base\_price = \etdestr\ tup\ \estr{sum\_base\_price}\ \ein$
  $\elet\ sum\_disc\_price = \etdestr\ tup\ \estr{sum\_disc\_price}\ \ein$
  $\elet\ sum\_charge = \etdestr\ tup\ \estr{sum\_charge}\ \ein$
  $\elet\ avg\_base\_price = \etdestr\ tup\ \estr{avg\_base\_price}\ \ein$
  $\elet\ avg\_disc\_price = \etdestr\ tup\ \estr{avg\_disc\_price}\ \ein$
  $\elet\ avg\_charge = \etdestr\ tup\ \estr{avg\_charge}\ \ein$
  $\etcons\ \estr{return\_flag}\ return\_flag\ ($
  $\etcons\ \estr{line\_status}\ line\_status\ ($
  $\etcons\ \estr{sum\_base\_price}\ sum\_base\_price\ ($
  $\etcons\ \estr{sum\_disc\_price}\ sum\_disc\_price\ ($
  $\etcons\ \estr{sum\_charge}\ sum\_charge\ ($
  $\etcons\ \estr{sum\_real\_cost}\ (sum\_base\_price * real\_cost)\ ($
  $\etcons\ \estr{sum\_margin}\ (sum\_base\_price * margin)\ ($
  $\etcons\ \estr{avg\_base\_price}\ avg\_base\_price\ ($
  $\etcons\ \estr{avg\_disc\_price}\ avg\_disc\_price\ ($
  $\etcons\ \estr{avg\_charge}\ avg\_charge\ ($
  $\etcons\ \estr{avg\_real\_cost}\ (avg\_base\_price * real\_cost)\ ($
  $\etcons\ \estr{avg\_margin}\ (avg\_base\_price * margin)\ ($
  $\etnil))))))))))))\ \ein$

/* Constructs the list of grouping/sorting attributes */
$\elet\ group\_sort\_attr = \elam tup.$
  $\elcons\ (\etdestr\ tup\ \estr{l\_returnflag})\ ($
  $\elcons\ (\etdestr\ tup\ \estr{l\_linestatus})\ ($
  $\etnil))\ \ein$

/* Constructs the aggregate functions */
$\elet\ group\_agg = \elam tup.$
  $\elet\ extended\_price = \etdestr\ tup\ \estr{l\_extended\_price}\ \ein$
  $\elet\ discount = \etdestr\ tup\ \estr{l\_discount}\ \ein$
  $\elet\ tax = \etdestr\ tup\ \estr{l\_tax}\ \ein$
  $\elet\ disc\_price = extended\_price * (1 - discount)\ \ein$
  $\elet\ charge = disc\_price * (1 + tax)\ \ein$
  $\etcons\ \estr{sum\_base\_price}\ (\esum\ extended\_price)\ ($
  $\etcons\ \estr{sum\_disc\_price}\ (\esum\ disc\_price)\ ($
  $\etcons\ \estr{sum\_charge}\ (\esum\ charge)\ ($
  $\etcons\ \estr{avg\_base\_price}\ (\eavg\ extended\_price)\ ($
  $\etcons\ \estr{avg\_disc\_price}\ (\eavg\ disc\_price)\ ($
  $\etcons\ \estr{avg\_charge}\ (\eavg\ charge)\ ($
  $\etnil))))))\ \ein$

/* Main query */
$\oproject_{project\_list}(\osort_{group\_sort\_attr}(\ogroup_{group\_sort\_attr, group\_agg}(\oscan_{\edataref{lineitem}}())))$
\end{lstlisting}

On this example, my implementation of the exhaustive reduction strategy does not terminate in reasonable time due to a combinatorial explosion. Nevertheless, it eventually finds the reduced expression with minimal measure, i.e the original expression with all definitions and common expressions inlined. On the other hand, my implementation of the heuristic-based strategy quickly finds the same result for a fuel value $\Phi \geq 11$ (timings are given in Table \ref{table:timings}). Furthermore, the reduced expression is a direct translation of the SQL query above, which represents an argument in favor of the measure described in Section \ref{sec:measure} and used by both reduction strategies.

\subsection{Example of fragment grouping: dynamic queries}
\label{sec:ex:dynamic}

Consider now a simple website on which users can read small ads from people selling furniture, cars, etc. and in particular a page to browse through the offers. This page would consist mainly in a form with (i) a dropdown menu to optionally sort ads by date or price, (ii) a set of checkboxes to filter by category and (iii) an integer field with a default value to specify the number of results to display on the page. The corresponding query, fetching the results, would have to be built dynamically depending on the presence of filters and ordering. The following QIR code is a possible query representation of this logic, in which we assume that variables \texttt{is\_sorted}, \texttt{sort\_attr}, \texttt{cat\_list} and \texttt{limit} are provided by the application context and encode the presence of a sorting attribute, the list of selected categories and the number of results to display.

\begin{lstlisting}[language=C,basicstyle=\small]
/* Recursively constructs a list of OR of the selected categories */
$\elet\ make\_cat\_filter = \elam cat\_list.\ \elam tup.$
  $\elet\ \erec\ aux = \elam filter.\ \elam cat\_list.$
    $\eldestr\ cat\_list\ filter\ (\elam hd.\ \elam tl.\ aux\ ((\etdestr\ tup\ \estr{category} = hd)\ \eor\ filter)\ tl)\ \ein$
  $\elet\ aux2 = \elam cat\_list.$
    $\eldestr\ cat\_list\ \elnil\ (\elam hd.\ \elam tl.\ aux\ (\etdestr\ tup\ \estr{category} = hd)\ tl)\ \ein$
  $aux2\ cat\_list\ \ein$

/* Constructs the ordering attributes */
$\elet\ make\_order = \elam attr.\ \elam tup.$
  $\eif\ attr = \estr{price}\ \ethen\ \elcons\ (\etdestr\ tup\ \estr{price})\ \elnil$
  $\eelse\ \eif\ attr = \estr{date}\ \ethen\ \elcons\ (\etdestr\ tup\ \estr{timestamp})\ \elnil$
  $\eelse\ \elnil\ \ein$

/* Constructs the list of projected attributes */
$\elet\ project\_list = \elam tup.$
  $\etcons\ \estr{title}\ (\etdestr\ tup\ \estr{title})\ ($
  $\etcons\ \estr{description}\ (\etdestr\ tup\ \estr{description})\ ($
  $\etnil))\ \ein$

/* Base table */
$\elet\ ads = \oscan_{\edataref{ads}}()\ \ein$

/* After category filters */
$\elet\ ads\_filtered = \eldestr\ cat\_list\ ads\ (\elam hd.\ \elam tl.\ \oselect_{make\_cat\_filter\ cat\_list}(ads))\ \ein$

/* After (optional) ordering */
$\elet\ ads\_ordered = \eif\ is\_sorted\ \ethen\ \osort_{make\_order\ sort\_attr}(ads\_filtered)\ \eelse\ ads\_filtered\ \ein$

/* Main query */
$\oproject_{project\_list}(\otopk_{limit}(ads\_ordered))$
\end{lstlisting}

On this example, my implementation of the exhaustive reduction strategy does not terminate, since one can obtain infinitely many distinct reductions of the original expression by unfolding the recursive call of \texttt{aux} in \texttt{make\_cat\_filter}. Conversely, my implementation of the heuristic-based strategy quickly finds the result expression with minimal measure (timings are given in Table \ref{table:timings}).\bigskip

For instance, with \texttt{cat_list} set to $\elnil$, \texttt{is_sorted} set to \texttt{false} and \texttt{limit} set to \texttt{20} the result found with $\Phi \geq 15$ is

\noindent $\oproject_{\elam t.\ \etcons\ \estr{title}\ (\etdestr\ t\ \estr{title})\ (\etcons\ \estr{description}\ (\etdestr\ t\ \estr{description})\ \etnil)}($\\
\indent $\otopk_{20}(\oscan_{\edataref{ads}}()))$\bigskip

And with \texttt{cat_list} set to $\elcons\ \estr{housing}\ (\elcons\ \estr{cars}\ \elnil)$, \texttt{is_sorted} set to \texttt{true}, \texttt{sort\_attr} set to \texttt{date} and \texttt{limit} set to \texttt{30} the result found with $\Phi \geq 24$ is

\noindent $\oproject_{\elam t.\ \etcons\ \estr{title}\ (\etdestr\ t\ \estr{title})\ (\etcons\ \estr{description}\ (\etdestr\ t\ \estr{description})\ \etnil)}($\\
\indent $\otopk_{30}(\osort_{\elam t.\ \elcons\ (\etdestr\ t\ \estr{timestamp})\ \elnil}($\\
\indent \indent $\oselect_{\elam t.\ \etdestr\ t\ \estr{category} = \estr{cars}\ \eor\ \etdestr\ t\ \estr{category} = \estr{housing}}(\oscan_{\edataref{ads}}()))))$\bigskip

Again, the result expression corresponds clearly to a SQL query, and this is a second example of the effectiveness of the measure described in Section \ref{sec:measure}.

\subsection{Example of incompatibility: caching}
\label{sec:ex:caching}

Following on the previous example, consider a page on which an admin can detect if an unexperienced user has published the same ad twice. Assume there is a function \texttt{unexperienced(user\_id)} outside the query boundary, telling if a user is unexperienced. This function would be translated into a Truffle dependency, and represented in the QIR code by a Truffle reference \texttt{\etruffle{0}}. The following QIR code could correspond to the query used by the page.

\begin{lstlisting}[language=C,basicstyle=\small]
/* Constructs the list of projected attributes */
$\elet\ project\_list = \elam tup.\ \etcons\ \estr{user\_id}\ (\etdestr\ tup\ \estr{user\_id})\ \etnil\ \ein$

/* Constructs the join condition */
$\elet\ join\_cond = \elam tup1.\ \elam tup2.$
  $(\etdestr\ tup1\ \estr{title} = \etdestr\ tup2\ \estr{title})\ \eand$
  $\enot\ (\etdestr\ tup1\ \estr{ad\_id} = \etdestr\ tup2\ \estr{ad\_id})\ \ein$

/* Unexperienced users */
$\elet\ ads\_unex\_users = \oselect_{\elam tup.\ \etruffle{0}\ (\etdestr\ tup\ \estr{user_id})}(\oscan_{\edataref{ads}}())\ \ein$

/* Main query */
$\oproject_{project\_list}(\ojoin_{join\_cond}(ads\_unex\_users,ads\_unex\_users))$
\end{lstlisting}

This is one example of situation where reducing a redex (e.g., inlining \texttt{ads\_unex\_users}) is not beneficial for the database. In this case, both the implementations of the exhaustive strategy and of the heuristic-based strategy (for $\Phi \geq 1$) quickly return the correct answer (timings are given in Table \ref{table:timings}), that is

\noindent $\elet\ ads\_from\_unex\_users = \oselect_{\elam tup.\ \etruffle{0}\ (\etdestr\ tup\ \estr{user_id})}(\oscan_{\edataref{ads}}())\ \ein$\\
\noindent $\oproject_{\etcons\ \estr{user\_id}\ (\etdestr\ tup\ \estr{user\_id})\ \etnil}($\\
\indent $\ojoin_{(\etdestr\ tup1\ \estr{title} = \etdestr\ tup2\ \estr{title})\ \eand\ \enot\ (\etdestr\ tup1\ \estr{ad\_id} = \etdestr\ tup2\ \estr{ad\_id})}($\\
\indent \indent $ads\_from\_unex\_users, ads\_from\_unex\_users))$\bigskip

This corresponds to a computation where the Truffle runtime embedded in the database evaluates once \texttt{ads\_unex\_users}, stores the result on the database storage then passes the main query to the database engine. Once again this is an example of the effectiveness of the measure described in Section \ref{sec:measure}.\bigskip

\begin{table}[ht]
\centering
\begin{tabular}{ | l | c | c | c | c | c | c | c | c |} \hline
\textbf{Experiment}     & \multicolumn{2}{c|}{Section \ref{sec:ex:analytics}}
                        & \multicolumn{2}{c|}{Section \ref{sec:ex:dynamic} (1)}
												& \multicolumn{2}{c|}{Section \ref{sec:ex:dynamic} (2)}
												& \multicolumn{2}{c|}{Section \ref{sec:ex:caching}}\\ \hline
\textbf{Red. strategy}  & Exh. & Heur. & Exh. & Heur. & Exh. & Heur. & Exh. & Heur.\\ \hline\hline
\textbf{Timing}         & 56.63s     & 6ms       & -          & 7ms       & -          & 7ms       & 5ms        & 4ms      \\ \hline
\textbf{Optimal result} & \cmark     & \cmark    & -          & \cmark    & -          & \cmark    & \cmark     & \cmark   \\ \hline
\end{tabular}
\caption{Timings of experiments from Sections \ref{sec:ex:analytics}, \ref{sec:ex:dynamic} and \ref{sec:ex:caching}}
\label{table:timings}
\end{table}

\vspace{-1.4em}
\section{Related work}
\label{sec:related-work}

As stated in the introduction, traditional data-oriented application frameworks, such as PHP with the MySQL extension or Java with the JDBC driver, construct queries in the application code using string concatenation or statements, with the aforementioned impedence mismatch issue. Object-relational mappings (ORM), such as \emph{criteria queries} in Hibernate and ActiveRecord in Ruby on Rails, offer a syntactic abstraction to represent querying primitives with constructs of the application language, and provide features to automatically handle translations from the application language data model to the database data model an vice versa.\bigskip

Although they solve part of the problem, these solutions restrict the application code sent to the database to the provided querying primitives. This results in numerous application-database roundtrips when a function of the application code (a UDF) is called from inside a query, and forms a well-known performance bottleneck\cite{ferry}\cite{data-access}. A recent contribution\cite{qbs} focuses on analyzing the application code surrounding querying primitives to translate it into the database language when possible, and inline it in the query. But this detection is based on a list of handcrafted patterns in which the choices of application language and database are hardcoded assumptions.\bigskip

More advanced solutions include the Forward project\cite{forward} (on which I was working last year) and Microsoft LINQ\cite{linq}. Forward offers as application language SQL++, a declarative extension of SQL tailored for web programming, and evaluates with an in-memory query processor language features not supported by the database. Although the similarities between database and application languages in terms of data models and querying primitives enable a seamless integration of queries in the application language, this framework fails to push to the database application code snippets containing unsupported features, and therefore does not fully solves the UDF performance problem. Compared to this contribution, the QIR framework offers the choice of the application language, a clean isolation of this application language from the rest of the framework, formal semantics for operators and expressions, allowing a precise description of how queries integrated in the application code are detected, translated and optimized, as well as the possibility to evaluate any application code in the database, regardless of its capabilities.\bigskip

LINQ integrates querying primitives for Microsoft SQL server in the application languages part of the .NET framework (e.g., C\#, F\#, Visual Basic). Similarly to ORMs, it provides automatic mechanisms to translate from the application language data model to the database data model an the other way around. Since LINQ has no released specification, we rely on the formal description and optimality proofs of T-LINQ\cite{t-linq} to compare our contribution. This approach adds to the application language a small language of expressions similar to QIR's (variables, lambdas, applications, list and tuples) as well as a quotation/antiquotation system (to define the boundary between application code to evaluate in the application runtime and in the database query processor). Querying primitives correspond to patterns in application code (e.g., a $\oscan$ corresponds to a \texttt{for} loop and a $\oselect$ to an \texttt{if} conditional. The translation to SQL includes a rewriting phase corresponding to successive iterations, on the application code AST, of the $\beta$-reduction and ad-hoc rules on these patterns. Although T-LINQ successfully gathers application code around the querying primitives and reduces it to produce efficient queries, it can only rely on the SQL capabilities of the database (i.e., there is no equivalent of the Truffle runtime inside the database). Thus, this framework only targets a limited number of application language and a single database, and conflates the three following issues: detection of the query boundary, optimization of the query representation and translation of QIR into the database language. This fact is particularly observable in the restrictions enforced on expressions inside quotes (i.e., to be sent to the database): quoted expressions have to be reductible to expressions supported by the database. Compared to this contribution, the QIR framework offers a clean separation between application language, query representation and database language. Moreover, it does not impose restrictions on application code that can be send to the database. One notices that, although they start from different architecture assumptions and use different proof ideas, Theorem \ref{th:completeness} ended up requiring hypotheses similar to the restrictions necessary to the optimality proof of T-LINQ\cite{t-linq}, hinting that such assumptions are reasonable. Nevertheless, while T-LINQ cannot handle expressions outside this scope (indeed, no example from Section \ref{sec:concrete-examples} can be reproduced in T-LINQ), the QIR heuristic can still produce (good) results on them, thanks to the embedded Truffle runtime. Looking at T-LINQ examples in \cite{t-linq}, one can easily spot the differences between the two frameworks: examples 1-4 and 6-7 can easily be reproduced in QIR and similar examples have already been given in this document, examples 5 and 8-9 contain nesting and unnesting operators and use ad-hoc reduction rules to rewrite them into a flat relational algebra (which we do not want to engage into, as a design choice), and examples 10-13 have to use the application language for features such as recursion, whereas QIR offers them natively.

\section{Conclusion and future work} 

The ability to evaluate snippets from the application code inside the database opened new perspectives on the language-intergrated query problems. We introduced Query Intermediate Representation (QIR), an abstraction layer between application language and database language, that can be used to represent queries. QIR comes with a rewriting strategy transforming queries written in a natural way in the application code into queries heavily optimizable by the database. Such rewritings are guided by a notion of ``good'' queries abstracted away in a measure on QIR expressions. Moreover, we produced formal proofs and experiments based on our QIR prototype showing that the QIR framework causes a minimal overhead while enabling a considerable speedup of query evaluation.\bigskip

Nevertheless, as explained in Section \ref{sec:architecture}, we made the assumption that a same construct has identical semantics in the applicaton language, QIR and the database language (e.g., a tuple navigation \texttt{user.name} returns the same result in JavaScript, QIR and SQL), and we allow a lossy translation from a language to another when this is not the case (in fact, exisiting frameworks have the same limitation). This issue, which extends to discrepancies in the data model (e.g., JavaScript integers and Cassandra integers might be different), could be fixed by introducing more abstraction in the QIR data model. On a different matter, the rewriting heuristic presented in Section \ref{sec:qir-evaluation} uses a fuel value that has to be set by the developer. This obligation could become a simple option if we provide a module ``guessing'' what fuel value to set, using the estimation logic discussed at the end of Section \ref{sec:qir-evaluation} and described in Appendix \ref{ap:fuel}. These ideas are future work. We also plan to try our framework with more combinations of application languages and databases, then measure the impact of such an architecture on real-world large-scale web applications.

\clearpage

\begin{appendices}

\section{References}
\label{ap:references}

\begingroup
\renewcommand{\section}[2]{}
\bibliographystyle{plain}
\bibliography{main}
\endgroup
\section{Internship timeline}
\label{ap:timeline}

The Gantt diagram of Figure \ref{fig:gantt} describes the use of my time during this internship. In addition to the elements presented in this report, I had the opportunity to present my work at Oracle Labs in July. It was a very interesting experience, a good presentation exercice and a chance to organize my results before writing this report.

\begin{figure}[ht]
\begin{center}
   \includegraphics[width=\linewidth]{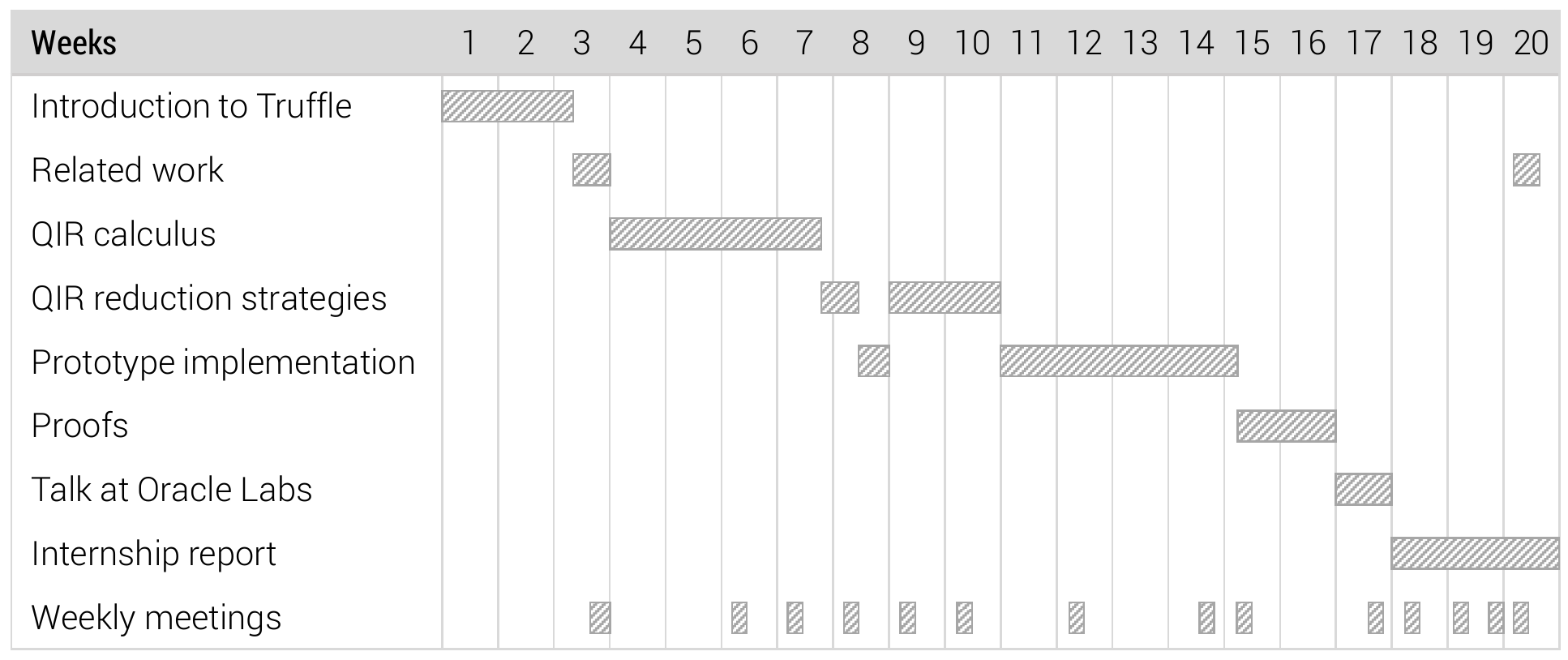}
	 \caption{Internship timeline}
	 \label{fig:gantt}
\end{center}
\end{figure}
\section{Formal definition of occurences}
\label{ap:qir-star}

Definitions \ref{def:qir*} and \ref{def:qir*-qir} correspond to a formalization of the concept of occurences presented quickly in Definition \ref{def:occurences}.

\begin{definition}
\label{def:qir*}
We define the language QIR* as the language of QIR constructs (cf. Section \ref{sec:qir-constructs}) in which integer identifiers can be attached to each subexpression and reductions are extended accordingly. Non-exhaustively, this means that 
\begin{itemize}
	\item QIR* operators contain $\oscan_{e_1}()$, $\oscan^{id}_{e_1}()$, $\oselect_{e_1}(e_2)$, $\oselect^{id}_{e_1}(e_2)$, etc.
	\item QIR* expressions contain $x$, $x^{id}$, $\elam x.\ e_1$, $\elam^{id} x.\ e_1$, $(e_1\ e_2)$, $(e_1\ e_2)^{id}$, $\etrue$, $\etrue^{id}$ etc.
	\item $x^{id}\{e_1/x\} = e_1$, $x\{e_1/x\} = e_1$
	\item $(\elam^{id} x.\ e_1)\ e_2 \rightarrow e_1\{e_2 / x\}$, $\eldestr^{id} \ (\elcons^{id'} \ e_1 \ e_2) \ e_3 \ e_4 \rightarrow e_4 \ e_1 \ e_2$, $\etrue^{id'}\ \eand^{id}\ e_1 \rightarrow\ \etrue$, etc.
\end{itemize}
where $id$, $id'$ are integer identifiers and the $e_i$ are QIR* expressions.\\
Let $e$ be a QIR expression and $e_1$ a subexpression of $e$. We denote by $\starify{e}$ the QIR* expression in which a unique identifier is attached to each subexpression of $e$. There is a one-to-one correspondance between the subexpressions of $e$ and the subexpressions of $\starify{e}$, thus we will also denote by $\starifysub{e_1}$ the subexpression of $\starify{e}$ corresponding to $e_1$.\\
Let $e'$ be a QIR* expression and $e'_1$ a subexpression of $e'$. We denote by $\unstarify{e'}$ the QIR expression in which all identifers are removed from $e'$. There is a one-to-one correspondance between the subexpressions of $e'$ and the subexpressions of $\unstarify{e'}$, thus we will also denote by $\unstarifysub{e'_1}$ the subexpression of $\unstarify{e'}$ corresponding to $e'_1$.\\
Let $e'$ be a QIR* expression and $e'_1$ a subexpression of $e'$. We denote by $\exprid{e'}{e'_1}$ the identifier of $e'_1$ in $e'$, if it has one (it is not defined otherwise).
\end{definition}

\begin{restatable}{lemma}{lemqirqirsimul}
\label{lem:qir*-qir-simul}
Let $e_0$ be a QIR expression and $e'_0$ a QIR* expression such that $\unstarify{e'_0} = e_0$. Any QIR reduction chain $e_0 \reduceredex{r_1} \ldots \reduceredex{r_n} e_n$ can be simulated by the QIR* reduction chain $e'_0 \reduceredex{r'_1} \ldots \reduceredex{r'_n} e'_n$ such that, for all $1 \leq i \leq n$, $\unstarify{e'_i} = e_i$ and $\unstarifysub{r'_i} = r_i$.
\end{restatable}

\begin{definition}
\label{def:qir*-qir}
Let $e_0$ be a QIR expression, $e$ a subexpression of $e_0$ and $e_n$ a QIR expression such that $e_0 \rightarrow^* e_n$. Using Lemma \ref{lem:qir*-qir-simul} and the fact that $\unstarify{\starify{e_0}} = e_0$, we know that this QIR reduction chain can be simulated by a QIR* reduction chain $e'_0 \rightarrow^* e'_n$ starting from $e'_0 = \starify{e_0}$. This way, we define $\occur{e_n}{e}$ as the set $\{ \unstarifysub{e'}\ |\ e'\ \text{is a subexpression of}\ e'_n, \exprid{e'_n}{e'} = \exprid{\starify{e_0}}{\starifysub{e}} \}$.
\end{definition} 

For instance, the QIR reduction\\
\noindent $(\elam x.\ \elam y.\ f\ x\ x)\ z\ z \rightarrow (\elam y.\ f\ z\ z)\ z \rightarrow f\ z\ z$\\
can be simulated by the QIR* reduction\\
\noindent $(\elam^1x.\ \elam^2y.\ f^3\ x^4\ x^5)\ z^6\ z^7 \rightarrow (\elam^2y.\ f^3\ z^6\ z^6)\ z^7 \rightarrow f^3\ z^6\ z^6$.\\
\indent One can see, just by looking at identifiers, that the first $z$ variable produced both $z$ variables of the final expression: it has two occurences in the final expression. The other one disappeared: it has no occurence in the final expression.
\section{Proofs}
\label{ap:proofs}

\thconfluence*
\begin{proof}
Encoding (i) integer (resp. boolean, string) values and functions in the lambda-calculus (Church encoding), (ii) operators, built-in functions, Truffle nodes, Data references, $\elnil$, $\elcons$, $\etnil$ and $\etcons$ with constructors and (iii) $\eldestr$ and $\etdestr$ with (linear) pattern-matching, QIR can be transposed in a lambda-calculus with patterns\cite{lambda-calc-patterns} which satisfies the Church-Rosser property and for which the $\beta$, $\delta$ and $\rho$ reduction rules of Definition \ref{def:reduction-rules} are compatible.
\end{proof}
\vspace{-0.6em}

\thstandardization*
\begin{proof}
With the same encoding as in the proof of Theorem \ref{th:confluence}, QIR can be transposed in a lambda-calculus with patterns\cite{standardization-patterns} in which the Standardization theorem holds and for which the $\beta$, $\delta$ and $\rho$ reduction rules of Definition \ref{def:reduction-rules} are compatible.
\end{proof}
\vspace{-0.6em}

\lemwellfounded*
\begin{proof}
The order induced by $M$ is a lexicographical order on the natural order of positive integers, which is well-founded, and the lexicographical order preserves well-foundness.
\end{proof}
\vspace{-0.6em}

\lemredsnonempty*
\begin{proof}
Since $e \in \reduction{0}{e} \subseteq \reductions{e}$, we know that $\reductions{e} \neq \emptyset$. Moreover, the order induced by $M$ is well-founded (Lemma \ref{lem:well-founded}) therefore $M$ has a minimum on $\reductions{e}$ and $\minreductions{e}$ is non-empty.
\end{proof}
\vspace{-0.6em}

\lemhinlinevarhmakeredexhcontractlamterminate*
\begin{proof}
The proof follows an induction on $\travpos{e}{e'}$.
\begin{itemize}
	\item If $e'$ is free in $e$ then $\hinlinevar{e}{e'}$ returns $\hnone$. Otherwise, let $l$ be the lambda binding $e'$. $\hinlinevar{e}{e'}$ corresponds to $\hcontractlam{e}{l}$, and since $\travpos{e}{l} < \travpos{e}{e'}$, the induction hypothesis yields the result.
	\item If $e'$ is a redex, $\hmakeredex{e}{e'}$ contracts this redex. If $e' = (e''\ e_1)$, $e' = \eldestr\ e''\ e_1\ e_2$, $e' = \etdestr\ e''\ s$ or $e' = \eif\ e''\ \ethen\ e_1\ \eelse\ e_2$ (and similarly for other $\rho$-redexes), $\hmakeredex{e}{e'}$ corresponds to $\hmakeredex{e}{e''}$, and since $\travpos{e}{e''} < \travpos{e}{e'}$ and the induction hypothesis yields the result. If $e'$ is a variable, $\hmakeredex{e}{e'}$ corresponds to $\hinlinevar{e}{e'}$ and we proved the result above. Otherwise, $\hmakeredex{e}{e'}$ returns $\hnone$.
	\item We denote by $e''$ the parent expression of $e'$. If $e'' = (e'\ e_1)$, $e'' = \eldestr\ e'\ e_1\ e_2$, $e'' = \etdestr\ e'\ s$ or $e'' = \eif\ e'\ \ethen\ e_1\ \eelse\ e_2$ (and similarly for other $\rho$-redexes), $\hcontractlam{e}{e'}$ corresponds to $\hmakeredex{e}{e''}$ which either contracts $e''$ and we are done, or corresponds to $\hmakeredex{e}{e'}$ and we proved the result above. If $e'' = (e_1\ e')$, $e'' = \eldestr\ e_1\ e'\ e_2$, $e'' = \eldestr\ e_1\ e_2\ e'$, $e'' = \eif\ e_1\ \ethen\ e'\ \eelse\ e_2$ or $e'' = \eif\ e_1\ \ethen\ e_2\ \eelse\ e'$ (and similarly for other $\rho$-redexes), $\hcontractlam{e}{e'}$ corresponds to $\hmakeredex{e}{e_1}$ and since $\travpos{e}{e_1} < \travpos{e}{e'}$, the induction hypothesis yields the result. If $e'' = \elam x.\ e'$, $e'' = \elcons\ e'\ e_1$, $e'' = \elcons\ e_1\ e'$, $e'' = \etcons\ s\ e'\ e_1$, $e'' = \etcons\ s\ e_1\ e'$, $\hcontractlam{e}{e'}$ corresponds to $\hcontractlam{e}{e''}$ and since $\travpos{e}{e''} < \travpos{e}{e'}$, the induction hypothesis yields the result. Otherwise, $\hcontractlam{e}{e'}$ returns $\hnone$.
\end{itemize}
\vspace{-1.2em}
\end{proof}
\vspace{-0.6em}

\lemhredcfghredchildterminate*
\begin{proof}
Each recursive call to $\hredcfg{\phi}{e}{C[]}$ (resp. $\hredchild{\phi}{e}{C[]}$) is well defined (corollary of Lemma \ref{lem:hinlinevar-hmakeredex-hcontractlam-terminate}) and is such that the pair $(M(e),\phi)$ decreases. Using lemma \ref{lem:well-founded} it is easy to prove that the lexicographical order induced by this pair is also well-founded, therefore the search space has bounded depth. Moreover, any expression has a bounded number of operators and operator contexts therefore each node in the search space has a bounded number of children.
\end{proof}
\vspace{-0.6em}

\lemqirqirsimul*
\begin{proof}
The proof follows an induction on $n$. For $n = 0$, the result is trivial. For $n > 0$, contracting the redex $r'_1$ such that $\unstarifysub{r'_1} = r_1$ yields a QIR* expression $e'_1$. It is easy to prove by case analysis on $r_1$ that $\unstarify{e'_1} = e_1$. By induction hypothesis, the remaining of the reduction chain $e_1 \rightarrow^* e_n$ can be simulated by $e'_1 \rightarrow^* e'_n$, which means that the entire chain $e_0 \rightarrow^* e_n$ can be simulated by $e'_0 \rightarrow^* e'_n$.
\end{proof}
\vspace{-0.6em}

\lemerasedorconsumed*
\begin{proof}
By case analysis on $r$. All cases are similar therefore we will only write the proof for $r = (\elam x. e_1)\ e_2$. By hypothesis, there is a context $C[]$ such that $e = C[r]$. Then we do a case analysis on the position of $e'$ in $e$. We know that $e'$ cannot be a subexpression of $C$, otherwise $\occur{e''}{e'} \neq \emptyset$. For the same reason, we also know that $e'$ is not a subexpression of $e_1$ other than the free variable $x$. If $e' = r$, if $e' = \elam x. e_1$ or if $e'$ is the variable $x$ free in $e_1$, then $e$ is consumed. Otherwise, $e'$ is a subexpression of $e_2$ and $x$ is not free in $e_1$ (since $\occur{e''}{e'} = \emptyset$), and in this case $e'$ is erased.
\end{proof}
\vspace{-0.6em}

\leminlinevarmakeredexcontractlamcomplete*
\begin{proof}
The proof follows an induction on $\travpos{e}{e'}$.
\begin{itemize}
	\item Suppose that $\hinlinevar{e}{v}$ returns $\hnone$. If $v$ is free in $e$ then by definition it cannot be consumed. Otherwise, denoting by $l$ the lambda binding $v$, this means that $\hcontractlam{e}{l}$ returns $\hnone$. Since $\travpos{e}{l} < \travpos{e}{v}$, the induction hypothesis tells that $l$ cannot be consumed and (by definition of the consumption of a variable) neither can $e'$.
	\item Suppose that $\hmakeredex{e}{e'}$ returns $\hnone$. If $e'$ is a variable, this implies that $\hinlinevar{e}{e'}$ returns $\hnone$ and we proved the result above. If $e'$ is a redex, this is absurd ($\hmakeredex{e}{e'}$ cannot return $\hnone$). If $e' = (e''\ e_1)$, $e' = \eldestr\ e''\ e_1\ e_2$, $e' = \etdestr\ e''\ s$ or $e' = \eif\ e''\ \ethen\ e_1\ \eelse\ e_2$ (and similarly for other $\rho$-redexes), this means that $\hmakeredex{e}{e''}$ returns $\hnone$ and since $\travpos{e}{e''} < \travpos{e}{e'}$ the induction hypothesis tells that $e''$ cannot be consumed and (by definition of the consumption of the considered expressions for $e'$) neither can $e'$. If $e' = \elam x.\ e_1$, $e' = \elcons\ e_1\ e_2$, $e' = \etcons\ s\ e_1\ e_2$, $e' = \elnil$, $e' = \etnil$ or if $e'$ is a constant (integer, boolean, etc.), by looking at all the call sites of $\hmakeredex{.}{.}$ we know that the parent of $e'$ is such that $e'$ is in an ill-formed term (e.g., $\eldestr\ \etnil\ e_1\ e_2$) and therefore cannot be consumed. The remaining possible expressions for $e'$ (e.g., Truffle nodes, Data references, operators, etc.) cannot be consumed in any expression.
	\item Suppose that $\hcontractlam{e}{e'}$ returns $\hnone$. We denote by $e''$ the parent expression of $e'$. If $e'' = (e'\ e_1)$, $e'' = \eldestr\ e'\ e_1\ e_2$, $e'' = \etdestr\ e'\ s$ or $e'' = \eif\ e'\ \ethen\ e_1\ \eelse\ e_2$ (and similarly for other $\rho$-redexes), this means that $\hmakeredex{e}{e''}$ returns $\hnone$, which in turn means that $\hmakeredex{e}{e'}$ returns $\hnone$ and we proved the result above. If $e'' = (e_1\ e')$, $e'' = \eldestr\ e_1\ e'\ e_2$, $e'' = \eldestr\ e_1\ e_2\ e'$, $e'' = \eif\ e_1\ \ethen\ e'\ \eelse\ e_2$ or $e'' = \eif\ e_1\ \ethen\ e_2\ \eelse\ e'$ (and similarly for other $\rho$-redexes), this means that $\hmakeredex{e}{e''}$ returns $\hnone$, which in turn means that $\hmakeredex{e}{e_1}$ returns $\hnone$. Since $\travpos{e}{e_1} < \travpos{e}{e'}$, the induction hypothesis tells that $e_1$ cannot be consumed and (by definition of the consumption of the considered expressions for $e''$) neither can $e''$, which implies that $e'$ cannot be consumed either. If $e'' = \elam x.\ e'$, $e'' = \elcons\ e'\ e_1$, $e'' = \elcons\ e_1\ e'$, $e'' = \etcons\ s\ e'\ e_1$, $e'' = \etcons\ s\ e_1\ e'$, this means that $\hcontractlam{e}{e''}$ returns $\hnone$ and since $\travpos{e}{e''} < \travpos{e}{e'}$, the induction hypothesis tells that $e''$ cannot be consumed and (by definition of the consumption of the considered expressions for $e'$) neither can $e'$. Similarly to above, the remaining expressions to consider for $e''$ either correspond to ill-formed expressions (that cannot be consumed) or expressions that can never be consumed.
\end{itemize}
\vspace{-1.2em}
\end{proof}
\vspace{-0.6em}

\lemreducedecrmeasure*
\begin{proof}
By case analysis on $r$. All cases are similar therefore we will only write the proof for $r = (\elam x. e_1)\ e_2$. By hypothesis, there is a context $C[]$ such that $e = C[r]$. Since $e$ has fixed operators, there is a one-to-one correspondance between the operators of $e$ and the operators of $e'$. For each operator $op$ in $e$, denoting by $op'$ the corresponding operator in $e'$, the stable compatibility hypothesis tells us that if $op$ is compatible, then $op'$ is also compatible. Since no redex can create operators, this implies that $\op{e'} - \comp{e'} \leq \op{e} - \comp{e}$. The only case to treat is when $\op{e'} - \comp{e'} = \op{e} - \comp{e}$. Looking at the definition of fragments (cf. Definition \ref{def:fragment}), we see that there is only three ways to increase the number of fragments in $e$.
\begin{itemize}
	\item Duplicate an existing fragment. This cannot happen under the fixed operator hypothesis, since a fragment contains at least one operator.
	\item Create a new fragment by making an incompatible operator compatible. This cannot happen either. Indeed, if $r$ turns an incompatible operator into a compatible one, using the stable compatibility hypothesis, we know that all other compatible operators in $e$ are still compatible in $e'$ which contradicts $\op{e'} - \comp{e'} = \op{e} - \comp{e}$.
	\item Split an existing fragment into at least two fragments. This again, cannot happen. Indeed, let $F$ be a fragment in $e = C[(\elam x. e_1)\ e_2]$. By definition of a fragment, we only have to consider the following cases:
	\begin{itemize}
		\item if $F$ is a subexpression of $e_2$ or $C[]$, it is intact in $e'$.
		\item if $F$ is a subexpression of $e_1$, either $x$ is not free in $F$ and $F$ is not affected by the reduction or $x$ is in the configuration of some operators of $F$ and (stable compatibility) these operators stay compatible after reduction and $r$ cannot split $F$.
		\item if $r$ is in the fragment, it is necessarily in a configuration of an operator of the fragment which (stable compatibility) stays compatible after reduction, therefore $r$ cannot split $F$.
	\end{itemize}
\end{itemize}
\vspace{-1.2em}
\end{proof}
\vspace{-0.6em}

\lemnormalmeasuremin*
\begin{proof}
Since $e$ is weakly-normalizing, it has a normal form $e_N$. Suppose, to the contrary, that there exists $e'$ such that $e \rightarrow^* e'$ and $M(e') < M(e_N)$. Using Theorem \ref{th:confluence} (confluence) and the fact that $e_N$ is normal, we know there is a finite reduction chain $e' \rightarrow^* e_N$ and applying Lemma \ref{lem:reduce-decr-measure} on each reduction of the chain leads to a contradiction.
\end{proof}
\vspace{-0.6em}

\lemsamecompatoperators*
\begin{proof}
Using the fixed operator hypothesis, we know that there is a one-to-one correspondance between the operators of $e$ and the operators in any reduced form of $e$. Therefore, $\comp{e'} = \comp{e_{min}}$.\\
Suppose, to the contrary, that there exists an operator $op$ in $e$ such that $\occur{e'}{op} = \{ op' \}$, $\occur{e_{min}}{op} = \{ op_{min} \}$, $op'$ compatible and $op_{min}$ not compatible. Using Theorem \ref{th:confluence} (confluence), we know there is an expression $e''$ such that $e' \rightarrow^* e''$ and $e_{min} \rightarrow^* e''$. Using the stable compatibility hypothesis, $op$ is compatible in $e''$ and all operators compatible in $e_{min}$ stay compatible in $e''$, which contradicts the minimality of the measure of $e_{min}$.\\
Suppose now, to the contrary, that there exists an operator $op$ in $e$ such that $\occur{e'}{op} = \{ op' \}$, $\occur{e_{min}}{op} = \{ op_{min} \}$, $op'$ not compatible and $op_{min}$ compatible. The minimality of $e_{min}$ tells that all operators compatible in $e''$ are also compatible in $e_{min}$ and the stable compatibility hypothesis tells that each operator compatible in $e'$ are still compatible in $e''$, which contradicts the fact that $\comp{e'} = \comp{e_{min}}$.
\end{proof}

\conjcompleteness*
Looking at the proof of the theorem for strongly-normalizing expressions (Theorem \ref{th:completeness}), one notices that the strongly-normalizing hypothesis itself is only used three times: (i) to prove that iteratively calling $\hinlinevar{.}{.}$ on the free variables of an operator configuration terminates, (ii) to use Theorem \ref{th:standardization} on an operator configuration and (iii) to prove that iteratively calling $\hmakeredex{.}{.}$ on an operator child expression terminates. Case (ii) is not an issue, since Theorem \ref{th:standardization} only requires the expression to be weakly-normalizing. Case (iii) can be dealt with using the fixed operator hypothesis: $\hmakeredex{.}{.}$ contracts redex in the same order as the call-by-name reduction strategy and none of these redexes can be erased (because this would require to erase the parent operator), therefore Theorem \ref{th:standardization} yields the result. The only remaining issue is case (i): $\hinlinevar{.}{.}$ reduces redexes in an order that cannot be compared with the call-by-name reduction strategy, therefore Theorem \ref{th:standardization} cannot be used here. However, the intuition is that operator configurations in which one can inline variables infinitely have to be inside some kind of fixpoint, which is forbidden by the fixed operator hypothesis.
\section{Estimation of the heuristic fuel value}
\label{ap:fuel}

For now, we assume that the fuel value (cf. Section \ref{sec:heuristic}) for the heuristic is empirically chosen by the developer. Nevertheless, one can give a simple estimate with the following reasoning. In usual QIR query representations, query fragments are not passed as function arguments and the tree of operators is almost already built.\bigskip

In this case, inlining a variable in a configuration is most of the time a single-step heuristic reduction (e.g., the inlining of $x$ in $\elet\ x = \ldots\ \ein\ C[op_{C'[x]}]$) or a multi-step heuristic reduction in which each step inlines a variable (e.g., the inlining of $y$ in $(\elam x.\ \elam y.\ C[op_{C'[x,y]}])\ e_1\ e_2$). Thus, denoting by $V$ the maximum (on all operators) number of variables to inline in a configuration, we can estimate that after $V$ reduction steps all variables are inlined, and providing more fuel would not inline more variables.\bigskip

Furthermore, configuration are usually small pieces of code, that can reach a normal form in a few reduction steps. The most expensive configurations to reduce are then expressions containing recursions. Thus, denoting by $R$ the maximum number of recursive calls and $R'$ the maximum number of reductions used inside a single recursive call, we can estimate that after $RR'$ reduction steps a configuration is in normal form, and providing more fuel would not enable to consider more reduced expressions of this configuration.\bigskip

Finally, children of operators are most of the time variables (e.g., $\elet\ subquery = \ldots\ \ein\allowbreak C[op(subquery)]$) or conditionals (e.g., $C[op(\eif\ c\ \ethen\ subquery_1\ \eelse\ subquery_2)]$) and both can be reduced to an operator in a small number of heuristic reduction steps, denoted by $C$. Thus, since QIR operators have at most two children, we can estimate that after $2C$ reduction steps all the children of an operator are reduced to an operator, and providing more fuel would not enable to regroup more fragments.\bigskip

Therefore, given the two-pass behavior of the heuristic, we can estimate the fuel value required for a comprehensive exploration of the search space as $\max(V+RR',2C)$. Looking back at the complex examples of Section \ref{sec:concrete-examples}, we can compare the aforementioned estimation (calculated by hand) to the fuel value $\Phi$ effectively required to reach the optimal result.

\begin{center}
\begin{tabular}{ | l || c | c | c | c | c || c |} \hline
Example                  & $V$ & $R$ & $R'$ & $C$ & $\max(V+RR',2C)$ & $\Phi$ \\ \hline\hline
\ref{sec:ex:analytics}   & $1$ & $1$ & $10$ & $0$ & $11$             & $11$   \\ \hline
\ref{sec:ex:dynamic} (1) & $2$ & $3$ & $8$  & $2$ & $26$             & $24$   \\ \hline
\ref{sec:ex:dynamic} (2) & $2$ & $3$ & $8$  & $2$ & $26$             & $15$   \\ \hline
\ref{sec:ex:caching}     & $1$ & $0$ & $0$  & $1$ & $1$              & $1$    \\ \hline
\end{tabular}
\end{center}

Although using the $\Phi$ estimation would lead to the optimal result in all these cases, computing $\Phi$ requires to estimate the values for $V$, $R$, $R'$ and $C$, which is far from easy to automate (for $R$ it corresponds to estimating in how many steps a recursive program terminates). Nevertheless, the $\Phi \approx \max(V+RR',2C)$ formula can be useful to the application developer when setting $\Phi$, since $V$, $R$, $R'$ and $C$ correspond to characteristics indirectly observable in the application language, as opposed to the $\Phi$ black-box.
\section{Source code and other resources}
\label{ap:resources}

The implementation source code and other resources mentioned in this document are available for download at \url{http://vernoux.fr/data/internship-mpri/}. This folder contains\bigskip

\dirtree{%
.1 /.
.2 integration_survey.xlsx\DTcomment{Notes on existing language-integrated query frameworks}.
.2 qir_reduction.\DTcomment{Prototype source code}.
.3 README.txt\DTcomment{Installation instructions}.
.3 src\DTcomment{Source code}.
.4 compat_sql.ml\DTcomment{SQL capabilities description}.
.4 compat_sql.mli\DTcomment{Interface of \texttt{compat_sql.ml}}.
.4 lexer.mll\DTcomment{Lexing of QIR expressions}.
.4 parser.mly\DTcomment{Parsing of QIR expressions}.
.4 qir.ml\DTcomment{QIR constructs and utilities}.
.4 qir.mli\DTcomment{Interface of \texttt{qir.ml}}.
.4 qir_reduction.ml\DTcomment{Measure and reduction strategies}.
.4 qir_reduction.mli\DTcomment{Interface of \texttt{qir_reduction.ml}}.
.4 reduce.ml\DTcomment{Main program and command-line arguments parsing}.
.4 utils.ml\DTcomment{Utilities}.
.3 test\DTcomment{Test files and examples}.
.4 ads_view.txt\DTcomment{Input of the experiment of Section \ref{sec:ex:dynamic}}.
.4 analytics.txt\DTcomment{Input of the experiment of Section \ref{sec:ex:analytics}}.
.4 caching.txt\DTcomment{Input of the experiment of Section \ref{sec:ex:caching}}.
.4 linq_1.txt\DTcomment{Reproduction of an example from \cite{t-linq}}.
.4 linq_2.txt\DTcomment{Reproduction of an example from \cite{t-linq}}.
.4 linq_xpath.txt\DTcomment{Reproduction of an example from \cite{t-linq}}.
.4 nesting_1.txt\DTcomment{Test case}.
.4 testchild_1.txt\DTcomment{Test case}.
.4 testchild_2.txt\DTcomment{Test case}.
.4 testchild_3.txt\DTcomment{Test case}.
.4 testconfig_1.txt\DTcomment{Test case}.
.4 testconfig_2.txt\DTcomment{Test case}.
.4 testconfig_3.txt\DTcomment{Test case}.
.4 testnoinline_1.txt\DTcomment{Test case}.
.4 testnoinline_2.txt\DTcomment{Test case}.
.4 testtpch_1.txt\DTcomment{Test case}.
}

\end{appendices}

\end{document}